\documentclass[11pt]{article}%
\usepackage{amsfonts}
\usepackage{amsmath}
\usepackage{amssymb}
\usepackage{amsthm}
\usepackage{graphicx}
\usepackage{fullpage,hyperref}%
\providecommand{\U}[1]{\protect\rule{.1in}{.1in}}
\newtheoremstyle{example}{\topsep}{\topsep}
{}
{}
{\bfseries}
{}
{  }
{\thmname{#1}\thmnumber{ #2}\thmnote{ (#3)}}

\newtheorem{theorem}{Theorem}

\newtheorem{corollary}[theorem]{Corollary}

\newtheorem{definition}[theorem]{Definition}

\newtheorem{proposition}[theorem]{Proposition}
\newtheorem{remark}[theorem]{Remark}

\theoremstyle{example}
\newtheorem{example}[theorem]{Example}

\let\originalleft\left
\let\originalright\right
\def\left#1{\mathopen{}\originalleft#1}
\def\right#1{\originalright#1\mathclose{}}

\begin{document}

\title{\textbf{Quantum enigma machines and the locking capacity}\\\textbf{of a quantum channel}}
\author{Saikat Guha\thanks{Quantum Information Processing Group, Raytheon BBN
Technologies, Cambridge, Massachusetts 02138, USA}
\and Patrick Hayden\thanks{Department of Physics, Stanford University, 382 Via Pueblo Mall, Stanford, California 94305-4060, USA}
\and Hari Krovi\footnotemark[1]
\and Seth Lloyd\thanks{Department of Mechanical Engineering, Massachusetts
Institute of Technology, Cambridge, MA 02139, USA}
\and Cosmo Lupo\thanks{Research Laboratory of Electronics, Massachusetts Institute
of Technology, Cambridge, MA 02139, USA}
\and Jeffrey H. Shapiro\footnotemark[4]
\and Masahiro Takeoka\thanks{National Institute of Information and Communications
Technology, 4-2-1 Nukuikita, Koganei, Tokyo 184-8795, Japan} $^{\ ,}$
\footnotemark[1]
\and Mark M. Wilde\thanks{Hearne Institute for Theoretical Physics,
Department of Physics and Astronomy, Center for
Computation and Technology, Louisiana State University, Baton Rouge, Louisiana
70803, USA}}
\maketitle

\begin{abstract}
The locking effect is a phenomenon which is unique to quantum information
theory and represents one of the strongest separations between the classical
and quantum theories of information. The Fawzi-Hayden-Sen (FHS) locking
protocol harnesses this effect in a cryptographic context, whereby one party
can encode $n$ bits into $n$ qubits while using only a constant-size secret key.
The encoded message is then secure against any measurement that an eavesdropper
could perform in an attempt to recover the message, but the protocol
does not necessarily meet the composability requirements needed in
quantum key distribution applications.
In any case, the locking effect represents
an extreme violation of Shannon's classical
theorem, which states that information-theoretic security holds in the
classical case if and only if the secret key is the same size as the message.
Given this intriguing phenomenon, it is of practical interest to study the
effect in the presence of noise, which can occur in the systems of both the
legitimate receiver and the eavesdropper. This paper formally defines the
\textit{locking capacity} of a quantum channel as the maximum amount of locked
information that can be reliably transmitted to a legitimate receiver by
exploiting many independent uses of a quantum channel and an amount of secret
key sublinear in the number of channel uses. We provide general operational
bounds on the locking capacity in terms of other well-known capacities from
quantum Shannon theory. We also study the important case of bosonic channels,
finding limitations on these channels' locking capacity when coherent-state
encodings are employed and particular locking protocols for these channels
that might be physically implementable.

\end{abstract}

\section{Introduction}

The security of a cryptographic primitive can be assessed according to
different security criteria. Most modern cryptosystems are
\textit{computationally} secure---that is, their security relies on the
difficulty of breaking them in a reasonable amount of time given available
technologies. This is also the case for the \textit{enigma machines}, a family
of historical polyalphabetic ciphers in use during the earlier half of the
previous century---their security relied on the difficulty of uncovering
patterns hidden in pseudorandom sequences~\cite{Bruen}.

A stronger security criterion requires that an encrypted message is close to
being statistically independent of the corresponding unencrypted message, in
which case one speaks of \textit{information-theoretic} security. For the case
of classical systems, a good measure of correlation is the mutual information
between the unencrypted and encrypted message. If the mutual information
vanishes, the chance of successfully decrypting the message is exponentially
small in the length of the message.
Any encryption scheme with such a property cannot perform any better than
one-time pad encryption, where a truly random key is used to encrypt (and
decrypt) the message \cite{S49}. The one-time pad guarantees
information-theoretic security as long as the key is kept secret, it has the
same length as the message, and can be used only once. However, the fact that
the secret key should be the same length as the message imposes severe
practical limitations on the use of the one-time pad protocol.

On the other hand, it is now known that quantum mechanics gives a way
around these limitations. The \textit{locking effect} is a phenomenon which is
unique to quantum information theory \cite{DHLST04} and represents one of the
most striking separations between the classical and quantum theories of
information. It is responsible for important revisions to security definitions
for quantum key distribution \cite{KRBM07} and might even help to explain how
both unitarity could be preserved and most of the information leaking from an
evaporating\ black hole could be inaccessible until the final stages of
evaporation \cite{SO06,DFHL10}. Quantum data locking occurs when the
accessible information about a classical message encoded into a quantum state decreases by an amount
that is much larger than the number of qubits of a small subsystem that is discarded
\cite{DHLST04}. A device that realizes a quantum data locking protocol is called
a quantum enigma machine~\cite{QEM}.

Impressive locking schemes exist \cite{HLSW04,DFHL10,FHS11}.
Suppose that a sender and receiver share a constant number of secret key bits. Using
these secret key bits, they can then encode an $n$-bit classical message into
$n$ qubits such that an adversary who gains access to these $n$ qubits, but
who does not know the secret key, cannot do much better than to randomly guess
the message after performing an arbitrary measurement on
these $n$ qubits.

However, the cryptographic applications of quantum data
locking have to be \textquotedblleft taken with a grain of
salt,\textquotedblright\ as they are only applicable if the distribution of
the message is completely random from the perspective of the adversary.
Otherwise, the key size should increase by an amount necessary to ensure that the
distribution of the message becomes uniform.
Moreover, one might say that the strength of quantum data locking also exposes a weakness.
Indeed, as a small key is sufficient for encrypting a long message,
the leakage of a small part of the
secret key may allow an adversary to uncover a disproportionate amount of information. 
For this reason, any cryptographic primitive based on the locking effect
(called a locking scheme), does not necessarily guarantee composable security \cite{KRBM07}.
This also implies that quantum data locking cannot necessarily
be used for secure key distribution. 
The only exception is if the adversary has no option other than to perform
a collective measurement on the qubits in her possession just after she
receives them.

As stated above, Shannon proved that such a locking effect is impossible
classically \cite{S49}. That is, when using only classical resources, a sender
and receiver require a secret key whose size is proportional to the size of
the message in order for the eavesdropper to have a  negligible
amount of information about the encrypted message. Thus, after Shannon's
result, information scientists looked in a different direction in order to
determine ways for communication systems to provide secrecy in addition to
reliable transmission. In reality, all communication systems suffer from
physical-layer noise, and one might be able to determine the characteristics
of the noise to a legitimate receiver and to an untrusted eavesdropper. Such a
model is known as the wiretap channel \cite{W75}, and it is well known now
that if the noise to the eavesdropper is stronger than the noise to the
legitimate receiver, then it is possible to communicate error-free at a
positive rate such that the eavesdropper obtains a  negligible
amount of information about the messages being transmitted.

\section{Summary of results}

In this paper, we consider the performance of locking protocols in the
presence of noise, and as an important application, we consider locking
protocols for bosonic channels. There are two types of noise to consider in
any realistic locking protocol: that which affects the transmission to a
legitimate receiver and that which affects the eavesdropper's system. Both are
important to consider in any realization of a quantum enigma machine.

We begin in Section~\ref{QDL} by reviewing the locking effect and a recently
introduced quantum enigma
machine (QEM) from Ref.~\cite{QEM}. This QEM encodes a classical message into a
single-photon state spread over a collection of discrete modes and then
decodes it by direct photodetection. The encryption and decryption are
realized by applying and inverting, respectively, a single multi-mode passive
linear-optical unitary transformation, selected uniformly at random from a set
of such transformations.
Similar to historical enigma machines, QEMs can encrypt a long message using
an exponentially shorter secret key. However, unlike historical enigma
machines which were only computationally secure, quantum data locking implies
security in the sense that the outcomes of any eavesdropper measurement will be essentially independent of the message.

After the review, Section~\ref{sec:locking-capacity}
provides a formal definition of the locking capacity of a
quantum channel. In short, a locking protocol uses a quantum channel $n$ times
(for some arbitrarily large integer $n$) and has three requirements:

\begin{enumerate}
\item The receiver should be able to decode the transmitted message with an
arbitrarily small error probability.

\item The eavesdropper can recover only an arbitrarily small number of the
message bits after performing a quantum measurement on her systems.

\item The secret key rate is no more than sublinear in the number $n$ of
channel uses (for example, logarithmic in $n$).
\end{enumerate}
We define the locking capacity of a quantum channel to be the maximum rate at
which it is possible to lock classical information according to the above
requirements. 
Changing the systems to which the adversary has access
leads to different notions of locking capacity, and we distinguish the notions
by naming them 
the {\it weak locking capacity} and the {\it strong locking capacity}. 
The difference between the two is that, in the weak notion, the adversary is assumed to have 
access to only the channel environment, while, in the strong case, we allow her
access to the channel input.
We emphasize that when we use the term
(weak or strong) ``locking capacity'' without any other modifiers,
we refer to the locking capacity of a
quantum channel without additional resources such
as classical feedback. Most of the results reported
here correspond to such a forward locking capacity.
The locking capacity of channels with additional resources
such as classical feedback remains largely open. However, at the very least,
we can already say that quantum key distribution protocols provide lower bounds
on the locking capacity in this setting.

We then find operational bounds on the locking capacity in terms
of other well known capacities studied in quantum Shannon theory, and we find
other information-theoretic upper bounds on the locking capacity. We prove
that the locking capacity of an entanglement-breaking channel is equal to
zero, which demonstrates that a quantum channel should have some ability to
preserve entanglement in order for it to be able to lock information according
to the above requirements. We also show that any achievable locking rate is equal to
zero whenever a given locking protocol has a classical simulation.
Furthermore, we find a class of channels for which the weak locking capacity
is equal to both the private capacity and quantum capacity. Finally, we
discuss locking protocols for some simple exemplary channels.

Section~\ref{CSQEM} establishes several important upper bounds on the locking
capacity of channels when restricting to coherent-state encodings.
If it were possible to exploit coherent-state encodings to perform locking
at high rates,
then this would certainly turn the locking effect from an interesting
theoretical phenomenon into one with practical utility. However, we are able
to show that there are fundamental limitations on the locking capacity
when restricting to coherent-state encodings. In particular, we prove that
the \textquotedblleft strong\textquotedblright\ locking capacity of any
channel is no larger than $\log_{2}(e)$ locked bits per channel use whenever
the encoding consists of coherent states (where $e$ is the base for the
natural logarithm). We also prove that the \textquotedblleft
weak\textquotedblright\ locking capacity of a pure-loss bosonic channel is no
larger than the sum of its private capacity and $\log_{2}(e)$. 

In Section~\ref{weakQEM}, we discuss an explicit protocol 
that uses a pulse position modulation encoding of coherent states. We derive bounds on the
security and key efficiency of this coherent-state locking protocol and find that it
has qualitative features analogous to the single-mode quantum enigma machine
in the presence of linear loss.

Finally, Section~\ref{end} presents our conclusions, a discussion of the
scaling of the required physical resources, and open questions for future research.

\section{Notation}

We briefly review some notation that we use in the rest of the paper.
Let $\mathcal{B}\left(
\mathcal{H}\right)  $ denote the algebra of bounded linear operators acting on
a Hilbert space $\mathcal{H}$. The $1$-norm of an operator $X$ is
defined as%
\[
\left\Vert X\right\Vert _{1}\equiv\text{Tr}\{  \sqrt{X^{\dag}X}
\} .
\]
Let $\mathcal{B}\left(  \mathcal{H}\right)  _{+}$ denote the subset of
positive semidefinite operators (we often simply say that an operator is
\textquotedblleft positive\textquotedblright\ if it is positive
semidefinite). We also write $X\geq0$ if $X\in\mathcal{B}\left(
\mathcal{H}\right)  _{+}$. An\ operator $\rho$ is in the set $\mathcal{D}%
\left(  \mathcal{H}\right)  $\ of density operators if $\rho\in\mathcal{B}%
\left(  \mathcal{H}\right)  _{+}$ and Tr$\left\{  \rho\right\}  =1$. The
tensor product of two Hilbert spaces $\mathcal{H}_{A}$ and $\mathcal{H}_{B}$
is denoted by $\mathcal{H}_{A}\otimes\mathcal{H}_{B}$.\ Given a multipartite
density operator $\rho_{AB}\in\mathcal{H}_{A}\otimes\mathcal{H}_{B}$, we
unambiguously write $\rho_{A}=\ $Tr$_{B}\left\{  \rho_{AB}\right\}  $ for the
reduced density operator on system $A$.

A linear map $\mathcal{N}%
_{A\rightarrow B}:\mathcal{B}\left(  \mathcal{H}_{A}\right)  \rightarrow
\mathcal{B}\left(  \mathcal{H}_{B}\right)  $\ is positive if $\mathcal{N}%
_{A\rightarrow B}\left(  \sigma_{A}\right)  \in\mathcal{B}\left(
\mathcal{H}_{B}\right)  _{+}$ whenever $\sigma_{A}\in\mathcal{B}\left(
\mathcal{H}_{A}\right)  _{+}$. Let id$_{A}$ denote the identity map acting on
 $\mathcal{B}\left(  \mathcal{H}_{A}\right)$.
 A linear map $\mathcal{N}_{A\rightarrow B}$ is completely
positive if the map id$_{R}\otimes\mathcal{N}_{A\rightarrow B}$ is positive
for a reference system $R$ of arbitrary size. A linear map $\mathcal{N}%
_{A\rightarrow B}$ is trace-preserving if Tr$\left\{  \mathcal{N}%
_{A\rightarrow B}\left(  \tau_{A}\right)  \right\}  =\ $Tr$\left\{  \tau
_{A}\right\}  $ for all input operators $\tau_{A}\in\mathcal{B}\left(
\mathcal{H}_{A}\right)  $. If a linear map is completely positive and
trace-preserving, we say that it is a quantum channel or quantum operation.
For simplicity, we denote a quantum channel
 $\mathcal{N}: \mathcal{B}({\mathcal{H}}_A) \mapsto \mathcal{B}({\mathcal{H}}_B)$ simply as
$\mathcal{N}_{A\to B}$. Similarly, we denote an isometry $U: {\mathcal{H}}_A \mapsto {\mathcal{H}}_B \otimes {\mathcal{H}}_C$ simply as $U_{A\to BC}$.

The variational
distance between two probability distributions $p(x)$ and $q(x)$ is defined as
$$
\sum_x | p(x) - q(x)| .
$$
The trace distance between two quantum states
$\rho$ and $\sigma$ is defined as follows:
$$
\Vert \rho - \sigma \Vert_1 ,
$$
and it is a conventional measure used in quantum information theory to quantify
the distinguishability of two quantum states. Clearly, when the two states are commuting,
the trace distance is equal to the variational distance between the two probability
distributions corresponding to the eigenvalues of $\rho$ and $\sigma$.

The von Neumann entropy of a state $\rho\in\mathcal{D}(\mathcal{H}_A)$ is
given by $H(A)_{\rho} := - \text{Tr} \{\rho\log \rho\}$. Throughout
this paper we take the logarithm base $2$. For a tripartite
state $\rho_{ABC} \in \mathcal{D}(\mathcal{H}_{ABC})$,
the quantum mutual information and the conditional quantum mutual information
are respectively given by:
\begin{align*}
I(A;B)_\rho & \equiv H(A)_\rho + H(B)_\rho - H(AB)_\rho, \\
I(A;B|C)_\rho & \equiv I(A;BC)_\rho - I(A;C)_\rho,
\end{align*}
where $ H(A)_\rho$ denotes the von Neumann entropy of the reduced state $\rho_A$, for example.

\section{Review of quantum data locking}

\label{QDL}A quantum data locking scheme can be implemented by a set of
$|\mathcal{K}|$ unitary transformations $\{U_{k}\}_{k\in\mathcal{K}}$ acting
on a Hilbert space $\mathcal{H}_{M}$ of finite dimension $|\mathcal{M}|$
\cite{DHLST04,HLSW04,DFHL10,FHS11}. (For the moment, we restrict ourselves to
finite-dimensional Hilbert spaces, but
Definitions~\ref{def:weak-lock-protocol}\ and \ref{def:strong-lock-protocol}
appearing later on allow for encoding information
into infinite-dimensional Hilbert spaces.) Alice encodes
$|\mathcal{M}|$ equiprobable messages by means of a set of orthonormal states
$\{|m\rangle\}_{m\in\mathcal{M}}$ defining a standard basis in $\mathcal{H}%
_{M}$. The encryption is then made by applying a particular unitary $U_{k}$
with $k$ chosen uniformly at random from $\mathcal{K}$, and this unitary maps
a standard basis state $|m\rangle$ into a state $U_{k}|m\rangle$. The
label~$k$ identifies the choice of the basis and plays the role of a secret key.

It is helpful to consider a particular classical-quantum state when reasoning
about a quantum data locking protocol. For such a state, we have two classical
systems, the first associated with Alice's message and the second associated
with the secret key, and a quantum system $Q$ of dimension~$|\mathcal{M}|$
corresponding to the quantum-encoded message of Alice. This classical-quantum
state is given by the following density matrix:%
\begin{equation}
\rho_{MKQ}=\frac{1}{|\mathcal{M}||\mathcal{K}|}\sum_{m,k}|m,k\rangle\langle
m,k|_{MK}\otimes\left(  U_{k}|m\rangle\langle m|U_{k}^{\dag}\right)
_{Q}\,,\label{rhoAB}%
\end{equation}
where the sets $\{|m\rangle\}$ and $\{|k\rangle\}$ are comprised of
orthonormal states representing the message and the secret key, respectively.
The receiver Bob has access to the quantum system $Q$ and the key system $K$.
We assume that an eavesdropper Eve only has access to the quantum system $Q$
(for example, before it gets passed along to the receiver Bob). The classical
correlations between Alice's message $M$ and Bob's systems $K$ and $Q$ can be
quantified by the \textit{accessible information} \cite{S90}. This is defined
as the maximum classical mutual information that can be extracted by
performing local measurements on the bipartite state:%
\begin{equation}
I_{\mathrm{acc}}(M;KQ)_{\rho}=\max_{\mathcal{M}_{KQ\rightarrow Y}%
}I(M;Y)\,,\label{eq:bob-acc-info}%
\end{equation}
where the maximization is taken over local measurement maps $\mathcal{M}%
_{KQ\rightarrow Y}$ and $I(X;Y)=H(X)+H(Y)-H(XY)$ is the mutual information,
with $H(Z)$ denoting the Shannon entropy of the random variable~$Z$%
~\cite{Cover}.

The accessible information in (\ref{eq:bob-acc-info}) can never be larger than
$\log_{2}|\mathcal{M}|$, due to the bound $I(M;Y)\leq\log_{2}|\mathcal{M}|$
which holds for any random variable $Y$. A particular strategy for achieving
this upper bound is for Bob first to perform the controlled unitary $\sum
_{k}|k\rangle\langle k|_{K}\otimes(U_{k}^{\dag})_{Q}$, leaving the state%
\[
\frac{1}{|\mathcal{M}||\mathcal{K}|}\sum_{m,k}|m,k\rangle\langle
m,k|_{MK}\otimes|m\rangle\langle m|_{Q}.
\]
He then simply measures in the basis $\{|m\rangle\}$ to recover the message
$m$ perfectly, so that his accessible information is maximal, equal to
$\log_{2}|\mathcal{M}|$.

To assess the security of the communication, let us consider the accessible
information for a party Eve who does not have access to the secret key. We
consider the following reduced state:%
\begin{equation}
\rho_{MQ}=\frac{1}{|\mathcal{M}|}\sum_{m}|m\rangle\langle m|_{M}\otimes
\frac{1}{|\mathcal{K}|} \sum_{k} \left(  U_{k}|m\rangle\langle m|U_{k}^{\dag
}\right)  _{Q}\,, \label{rhoAE}%
\end{equation}
obtained by taking the partial trace over the key system in (\ref{rhoAB}).
The aim of Eve is to find an optimal positive operator-valued measure (POVM)
to maximize the classical mutual information. It is sufficient to consider a
POVM $\mathcal{M}_{Q\to Y}$ with rank-one measurement operators, i.e.,%
\begin{equation}
\{\mu_{y}|\phi_{y}\rangle\langle\phi_{y}|\}, \label{eq:rank-1-meas}%
\end{equation}
where each $|\phi_{y}\rangle$ is a normalized vector and $\mu_{y}>0$ (the
sufficiency of rank-one POVMs follows by a data processing argument). We then
find the following expression for Eve's accessible information about Alice's
message \cite{DHLST04}:%
\begin{equation}
I_{\mathrm{acc}}(M;Q)_{\rho}=\log_{2}{|\mathcal{M}|}-\min_{\mathcal{M}_{E\to
Y}}\sum_{y}\frac{\mu_{y}}{|\mathcal{M}||\mathcal{K}|}\sum_{k}H(q_{yk})\,,
\label{acc}%
\end{equation}
where the probability distributions $q_{yk}$ have components $q_{yk}%
^{m}=|\langle\phi_{y}|U_{k}|m\rangle|^{2}$. Notice that Eve's accessible
information is written in terms of the minimum of the Shannon entropies
$H(q_{yk})=-\sum_{m} q_{yk}^{m}\log_{2}{q_{yk}^{m}}$ averaged over $y$ and $k$.

While finding Eve's optimal POVM is generally a difficult problem, one can
obtain a good upper bound by a convexity argument \cite{DHLST04}.
Furthermore, one can choose the encoding unitaries uniformly at random
according to the Haar measure \cite{HLSW04,FHS11,Buhrman,DFHL10}, and
if one also adjoins to the message
a small ancilla system in a maximally mixed state \cite{FHS11}, then
 it is possible to reduce the adversary's
accessible information to become arbitrarily small. These latter results show
that for large enough $|\mathcal{M}|$ there exist data locking schemes with
$\log_{2}|\mathcal{K}|$ negligibly small in comparison to $\log_{2}%
|\mathcal{M}|$ and for which
\[
I_{\mathrm{acc}}(M;Q)_{\rho}\ll I_{\mathrm{acc}}(M;KQ)_{\rho}.
\]
That means that a relatively short secret key can be used to encrypt an
exponentially longer message. To be more precise, consider the results of
\cite{FHS11}, according to which for $\left\vert \mathcal{M}\right\vert$ large enough 
there exist choices of $|\mathcal{K}|$ unitaries, with
\begin{equation}
\log_{2}{|\mathcal{K}|} = 4\log_{2}{(\varepsilon^{-1})}
+ O(\log_2\log_{2}{(\varepsilon^{-1})})\,,
\end{equation}
such that
\begin{equation}
I_{\mathrm{acc}}(M;Q)_{\rho} \leq \varepsilon \log_{2}{|\mathcal{M}|} \,,
\end{equation}
for any $\varepsilon>0$.
Moreover, if one randomly chooses the
$|\mathcal{K}|$ unitaries according to the Haar distribution on the unitary
group, the probability of picking up a set with this property approaches one
exponentially fast in the limit as $\left\vert \mathcal{M}\right\vert
\rightarrow\infty$.

In quantum data locking, the removal of a subsystem reduces the accessible
information by an amount larger than the number of qubits removed. This is a
purely quantum feature which has no classical analog. For comparison,
consider a classical counterpart of the quantum data locking setting, in which
Alice has access to a message variable $M$, Bob to an output random
variable~$Y$ and key variable $K$, while Eve has access only to $Y$. In the
classical framework, the following inequality holds%
\begin{equation}
I(M;YK)-I(M;Y)=I\left(  M;K|Y\right)  \leq H(K)\leq\log_{2}{|\mathcal{K}|}\,.
\label{cl-ineq}%
\end{equation}
This inequality shows that in
the classical framework, removal of the key variable $K$ reduces the
mutual information by no more than $\log_{2}{|\mathcal{K}|}$.

In the quantum case as discussed above, this inequality can be violated by an
arbitrarily large amount by replacing the classical mutual information with
the accessible information. A violation of the classical inequality in
(\ref{cl-ineq})\ can be quantified in terms of the following ratios
\cite{DHLST04,HLSW04}:%
\begin{align}
r_{1}  &  =\frac{I_{\mathrm{acc}}(M;Q)_{\rho}}{I_{\mathrm{acc}}(M;KQ)_{\rho}%
}\,,\label{r1}\\
r_{2}  &  =\frac{\log_{2}{|\mathcal{K}|}}{I_{\mathrm{acc}}(M;KQ)_{\rho
}-I_{\mathrm{acc}}(M;Q)_{\rho}}\,. \label{r2}%
\end{align}
The first is the ratio of the accessible information without the secret key to
that with the secret key. The second is the ratio of the key length to the
amount of information that Bob can unlock by having access to the key.
For a good locking scheme, both of
these quantities should be small, and the quantum data locking schemes
discussed above are such that both $r_{1}$ and $r_{2}$ can be made arbitrarily
small. On the other hand, the inequality in (\ref{cl-ineq}) implies that
$r_{2}\geq1$ for any locking scheme that uses classical resources only. Notice
that the one-time pad protocol has $r_{2}=1$ because the number of bits in the
key is equal to the amount of unlocked information for Bob.

\subsection{Quantum enigma machine}\label{unaryQEM}

A particular example of a QEM was proposed in
Ref.~\cite{QEM}.
This QEM implements an optical realization of quantum data locking, in
which Alice exploits a pulse position modulation (PPM) encoding using
single-photon states over $n$ optical modes \cite{QEM}. The message states
$|m\rangle=a_{m}^{\dag}|0\rangle$ represent the states of a single photon
occupying one out of a set of $n$ bosonic modes with canonical operators
$\{a_{m},a_{m}^{\dag}\}_{m\in\{1,\dots,n\}}$. Thus, for this case, we have
$n=|\mathcal{M}|$. The unitaries $\{U_{k}\}_{k\in\mathcal{K}}$ are realized as
passive linear-optical transformations acting on $n$ modes. The encryption
through a passive linear-optical unitary $U_{k}$ transforms the message states
into
\begin{equation}
|m\rangle_{k}:=U_{k}|m\rangle=\sum_{m^{\prime}=1}^{n} \widetilde{U}%
_{k}^{\left(  m,m^{\prime}\right)  } |m^{\prime}\rangle\, ,
\end{equation}
where $\widetilde{U}_{k}$ is the corresponding $n\times n$ unitary matrix
acting on the mode labels. The effect of the encryption is to spread a single
photon coherently over $n$ modes.

Let us first assume that Alice and Bob communicate via a noiseless quantum
channel. Then Bob receives the state prepared by Alice unperturbed. He
decrypts the message by first applying the inverse transformation $U_{k}%
^{\dag}$ and then by performing photodetection on the modes
$\{a_{m}\}$. We assume that Eve may intercept the signal but she does not
know which unitary has been used for encryption. 
Then, a direct application of the results of \cite{FHS11} shows that Eve's accessible
information can be made arbitrary small using a pre-shared secret key of
length logarithmic in the length of the message.

One natural application of a QEM is in synergy with standard quantum key
distribution (QKD) \cite{BB84,RevModPhys.81.1301}---that is, a relatively
short secret key can be first established by QKD and then used to encrypt a
much (exponentially) longer message through the QEM. This combination of QKD
and QEM in an all-quantum-optical cryptosystem could possibly overcome the
bit-rate limitations of standard QKD, but more work is necessary to determine
if this is the case.

Let us now suppose that Alice and Bob communicate through a pure-loss bosonic
channel with transmissivity $\eta\in(0,1)$ and Eve makes a passive wiretap attack on the
communication line, hence getting the photon lost in the channel
with probability no larger than $1-\eta$.
A simple feedback-assisted strategy allows for Alice and Bob to use the same scheme even for transmissivity
values below $50\%$. Notice that the only effect of the pure-loss channel is
to induce a probabilistic leakage of the photon. Hence, each time Bob detects
a photon (which happens with probability $\eta$) he can be sure that he has
correctly decrypted Alice's message. On the other hand, if Bob's
photodetectors do not produce a click (which happens with probability $1-\eta
$) he can request for Alice to resend. This shows that with the help of a
classical feedback channel Alice and Bob can attain the accessible
information
\begin{equation}
I_{\mathrm{acc}}(M;KQ)_{\rho}=\eta\log_{2}{n}\,.
\end{equation}
Although reduced, this value of the accessible information equals the maximum
value achievable through a pure-loss bosonic channel with a mean value of
$n^{-1}$ photons per mode \cite{QEM,GiovPRL}. 

Lloyd argues that such a scheme should be secure in principle \cite{QEM}.
However, a critical assumption for this security to hold is that Eve should attack 
each block that she receives independently, in which case her accessible information 
is reduced by a factor $1-\eta$ when compared to the lossless case.
Indeed, an important assumption for the security of any locking
protocol is that the distribution of the message is uniform from the
perspective of the adversary. If this is not the case (as for repeated
transmission of the same message when it does not show up at the receiver's
end), then the secret key needs to be large enough so that the distribution of
the message becomes uniform (see Proposition~4.16 in Ref.~\cite{F12}).

Concerning the key efficiency of the protocol, we can estimate the key efficiency ratio as
\begin{equation}
r_{2} \simeq \frac{4\log_{2}{(\varepsilon^{-1}})}{\eta\log_{2}{n}}\,.
\end{equation}
This expression implies that, although $r_{2}$ can be made arbitrarily small
by increasing $n$, the number of bosonic modes needed to fulfill the key
efficiency condition $r_{2}<1$ grows exponentially with decreasing $r_{2}$ and $\eta$.
This feature is first of all a consequence of the fact that the quantum data locking scheme
in \cite{FHS11} (similar conclusions are also obtained using the results of \cite{HLSW04,DFHL10})
requires a high-dimensional Hilbert space. On top of that, there is the fact that
the PPM encoding, as remarked above, is highly inefficient as it encodes $\log_2{n}$ qubits into
$n$ optical modes.

According to Definition \ref{def:weak-lock-protocol} below, this QEM is an
instance of an $(n,R,\varepsilon)$ weak locking protocol (assisted by classical
feedback) for the pure-loss bosonic channel with transmissivity $\eta$, with a
locking rate $R=\left[  \eta\log_{2}{n}\right]/n$.
It is worthwhile to notice that, due to the inefficiency of PPM encoding,
the rate of this QEM approachs zero as $n$ increases.

\section{The locking capacity of a quantum channel}

\label{sec:locking-capacity}
In this section, we take a more general approach to quantum data locking than
that pursued in prior work by defining the \textit{locking capacity} of a
quantum channel. Our goal is to understand the locking effect in the setting
of quantum Shannon theory, where a sender and receiver are given access to
$n$\ independent uses of a noisy quantum channel (where $n$ is an arbitrarily large
integer). Their aim is to exploit some sublinear (in $n$) amount of secret key
in order to lock classical messages from an adversary, in the sense that this
adversary will not be able to do much better than random guessing
when performing a quantum measurement to
learn about the transmitted message. Also, we demand that
the legitimate receiver (who knows
the value of the secret key) be able to recover the classical message with an
arbitrarily small probability of error. This leads us naturally to the
following formal definition of a locking protocol for a noisy channel:

\begin{definition}
[Weak locking protocol]\label{def:weak-lock-protocol}An $\left(
n,R,\varepsilon\right)  $ weak locking protocol for a channel $\mathcal{N}%
_{A\rightarrow B}$ consists of encoding and decoding maps $\mathcal{E}%
_{MK\rightarrow A^{n}}$ and $\mathcal{D}_{B^{n}K\rightarrow\hat{M}}$,
respectively. The encoding $\mathcal{E}_{MK\rightarrow A^{n}}$\ acts on a
message system $M$ and a key system $K$ and outputs the system $A^{n}$ for
input to $n$ uses of the channel. The decoding map $\mathcal{D}_{B^{n}%
K\rightarrow\hat{M}}$ acts on the output systems $B^{n}$ and the key system
$K$ to produce a classical system $\hat{M}$ containing the receiver's estimate
of the message. Without loss of generality, the encoding consists of
$\left\vert \mathcal{M}\right\vert \left\vert \mathcal{K}\right\vert $ quantum
states $\rho_{m,k}$, where $\left\vert \mathcal{M}\right\vert $ is the number
of messages and $\left\vert \mathcal{K}\right\vert $ is the number of key
values. Furthermore, the decoding consists of $\left\vert \mathcal{K}%
\right\vert $ POVMs \thinspace$\{\Lambda_{m}^{\left(  k\right)  }%
\}_{m\in\mathcal{M}}$. 
The rate $R=\log_{2}\left\vert
\mathcal{M}\right\vert /n$ and the parameter $\varepsilon >0$.
The protocol should satisfy the following requirements:

\begin{enumerate}
\item Given the key, the receiver can decode the transmitted message well on average:%
\[
\frac{1}{\left\vert \mathcal{M}\right\vert \left\vert \mathcal{K}\right\vert
}\sum_{m,k}\operatorname{Tr}\left\{  \Lambda_{m}^{\left(  k\right)  }\left(
\mathcal{N}_{A\rightarrow B}\right)  ^{\otimes n}\left(  \rho_{m,k}\right)
\right\}  \geq1-\varepsilon.
\]

\item Let $\{\Gamma_y\}$ be a POVM that Eve can perform in an attempt to learn about
the message $M$. After she performs this measurement, the joint classical-classical
state of the message and her measurement outcome is as follows:
\[
\frac{1}{\left\vert \mathcal{M}\right\vert }\sum_{m}\left\vert m\right\rangle
\left\langle m\right\vert _{M}\otimes \sum_y \operatorname{Tr}
\left\{ \Gamma_y \left(\frac{1}{\left\vert \mathcal{K}%
\right\vert }\sum_{k}\left(  \mathcal{N}_{A\rightarrow E}\right)  ^{\otimes
n}\left(  \rho_{m,k}\right)\right)  \right\} |y\rangle\langle y|_Y,
\]
where $\mathcal{N}_{A\rightarrow E}$ is the channel complementary to
$\mathcal{N}_{A\rightarrow B}$.
Equivalently, the joint probability distribution $p_{M,Y}(m,y)$ is equal to
\[
p_{M,Y}(m,y) = \frac{1}{\left\vert \mathcal{M}\right\vert }
 \operatorname{Tr}
\left\{ \Gamma_y \left(\frac{1}{\left\vert \mathcal{K}%
\right\vert }\sum_{k}\left(  \mathcal{N}_{A\rightarrow E}\right)  ^{\otimes
n}\left(  \rho_{m,k}\right)\right)  \right\} ,
\]
Our security criterion (see also Ref.~\cite{FHS11})
is that, for any measurement outcome $y$ of Eve, the variational distance between
the message distribution $p_M(m)$ and the distribution $p_{M|Y}(m|y)$ for the message
conditioned on any particular measurement outcome should be no larger
than $\varepsilon$:
\begin{equation}
\sum_m |p_M(m) - p_{M|Y}(m|y) | \leq \varepsilon .
\label{eq:sec-crit-var-dist}
\end{equation}
The interpretation here is that Eve cannot do much better than to randomly
guess the message if all
of the conditional distributions $p_{M|Y}(m|y)$ are indistinguishable from
the message distribution.

\item The secret key consumption grows sublinearly in the number $n$ of channel uses.
\end{enumerate}
\end{definition}

In a weak locking protocol, it is assumed that the eavesdropper has access 
to the channel environment only. 
A stronger locking protocol is obtained if we allow for the eavesdropper to have access 
to the channel input (or, equivalently, to both the channel output and environment):

\begin{definition}
[Strong locking protocol]\label{def:strong-lock-protocol}An $\left(
n,R,\varepsilon\right)  $ strong locking protocol is similar to a weak locking
protocol, except that we allow for Eve to have access to the $A^n$ systems, so
that she can perform a measurement on the $A^n$ systems
of the following state:
\[
\frac{1}{\left\vert \mathcal{M}\right\vert }\sum_{m}\left\vert m\right\rangle
\left\langle m\right\vert _{M}\otimes\frac{1}{\left\vert \mathcal{K}%
\right\vert }\sum_{k}\left(  \rho_{m,k}\right)  _{A^{n}}.
\]
We then demand that the variational distance as in (\ref{eq:sec-crit-var-dist})
can be made less than
an arbitrarily small positive constant $\varepsilon$.

\end{definition}

\begin{remark}
One could alternatively allow for the adversary to have access to the output of the channel,
but we do not explore such a possibility in this paper. 
\end{remark}

\begin{remark}\label{rem:FA-acc-bound}
The Fannes-Audenaert inequality \cite{Fannes73,A07} for continuity of entropy
implies that if (\ref{eq:sec-crit-var-dist})
holds, then we get the following bound on Eve's accessible
information:
\begin{equation}
I_{\operatorname{acc}}\left(  M;E^{n}\right) \leq
h_2(\varepsilon/2) + \varepsilon nR/2,
\label{eq:acc-info-bound}
\end{equation}
where $h_2$ is the binary entropy and $n$ and $R$ are as in
Definition~\ref{def:weak-lock-protocol}.
In more detail, recall the Fannes-Audenaert inequality for continuity of entropy:
$$
T \equiv \tfrac{1}{2}\Vert \rho - \sigma \Vert_1  \implies
|H(\rho) - H(\sigma)| \leq h_2(T) + T \log (d-1) ,
$$
where $h_2$ is the binary entropy and $d$ is the dimension of the states.
Applying this to the condition in (\ref{eq:sec-crit-var-dist}) gives
\begin{align}
H(M) - H(M|Y=y) & \leq h_2(\varepsilon/2) + \tfrac{\varepsilon}{2}
\log(|\mathcal{M}|-1) \\
& \leq  h_2(\varepsilon/2) + \varepsilon
nR/2.
\end{align}
Since the above inequality holds for any measurement of Eve,
averaging it with respect to the distribution $p_Y(y)$ gives the inequality 
in (\ref{eq:acc-info-bound}).
\end{remark}

\begin{remark}
If desired, one can demand further for the secret key rate of an $\left(
n,R,\varepsilon\right)  $ weak or strong locking protocol to be consumed at a
particular sublinear rate (for example, a logarithmic number of secret key bits
 or perhaps $\sqrt{n}$ secret key bits for $n$ channel uses).
However, the present paper establishes several upper bounds
on locking capacity in an IID~setting, and these bounds converge to the same quantity
in the large $n$ limit regardless
of which sublinear rate is chosen. Also, the FHS protocol \cite{FHS11} is very strong,
in the sense that it uses such a small amount of secret key.
Thus, in light of these two observations
 it seems reasonable to define locking capacity in such a coarse-grained manner.
However, other characterizations of locking capacity in a finite blocklength setting 
or in a one-shot setting might change depending on the amount of secret key allowed
(so it would be necessary to specify in more detail the amount of secret key allowed).
\end{remark}

\begin{remark}
\label{rem:strong-implies-weak-lock}Observe that an $\left(  n,R,\varepsilon
\right)  $ strong locking protocol is also an $\left(  n,R,\varepsilon\right)
$ weak locking protocol, but the other implication is not necessarily true.
\end{remark}

\begin{remark}
The security and key efficiency ratios become arbitrarily small for a strong
locking protocol. Indeed, from the fact that $I_{\operatorname{acc}}\left(
M;A^{n}\right)  \leq h_2(\varepsilon/2) + \varepsilon
nR/2$ and the fact that the receiver can decode
with the key, so that $I_{\operatorname{acc}}\left(  M;B^{n}K\right)
\approx\log_{2}\left\vert \mathcal{M}\right\vert $, it follows that the
security ratio $r_{1}\leq h_2(\varepsilon/2)/(nR) + \varepsilon
/2$. Also, since we require the key to be
sublinear in the message length, it follows that the key efficiency ratio
$r_{2}=o\left(  n\right)  /O\left(  n\right)  $, which vanishes in the limit
as $n\rightarrow\infty$.
\end{remark}

In what follows, we use the modifier \textquotedblleft weak\textquotedblright%
\ or \textquotedblleft strong\textquotedblright\ only when we need to
distinguish between them.

\begin{definition}
[Achievable rate for locking]\label{def:achievable-rate-locking}A rate $R$ is
achievable if $\forall\,\delta,\varepsilon>0$ and sufficiently large $n$,
there exists an $\left(  n,R-\delta,\varepsilon\right)  $ locking protocol.
\end{definition}

\begin{definition}
[Locking capacity]\label{def:locking-cap}The locking capacity $L\left(
\mathcal{N}\right)  $\ of a quantum channel is the supremum of all achievable
rates:%
\[
L\left(  \mathcal{N}\right)  \equiv\sup\left\{  R\ |\ R\text{ is
achievable}\right\}  .
\]
Let $L_{W}\left(  \mathcal{N}\right)  $ and $L_{S}\left(  \mathcal{N}\right)
$ denote the weak and strong locking capacity, respectively.
\end{definition}

\subsection{Relation of the locking capacity to other capacities}

Let $Q\left(  \mathcal{N}\right)  $, $P\left(  \mathcal{N}\right)  $, and
$C\left(  \mathcal{N}\right)  $\ denote the quantum
\cite{PhysRevA.54.2614,PhysRevA.54.2629,BNS98,BKN98,PhysRevA.55.1613,capacity2002shor,ieee2005dev}%
, private \cite{ieee2005dev,1050633}, and classical
\cite{Hol98,PhysRevA.56.131}\ capacities of a quantum channel $\mathcal{N}$,
respectively. By employing operational arguments, we can determine that the
following bounds hold%
\begin{equation}
Q\left(  \mathcal{N}\right)  \leq P\left(  \mathcal{N}\right)  \leq
L_{W}\left(  \mathcal{N}\right)  \leq C\left(  \mathcal{N}\right)  .
\label{eq:weak-lock-bounds}%
\end{equation}
Indeed, for any channel, its quantum capacity is less than the private
classical capacity because any scheme for quantum communication can be used
for private classical communication such that the classical information is
protected from the environment of the channel. Furthermore, the inequality
$P\left(  \mathcal{N}\right)  \leq L_{W}\left(  \mathcal{N}\right)  $ holds
because any $\left(  n,R,\varepsilon\right)  $ private classical communication
protocol satisfies the three requirements of a weak locking protocol
\cite{ieee2005dev,1050633}. Finally, the requirements of a weak locking
protocol are more restrictive than those for classical communication, so that
$L_{W}\left(  \mathcal{N}\right)  \leq C\left(  \mathcal{N}\right)  $.

Operational arguments and the existence of the
Fawzi-Hayden-Sen\ (FHS)\ locking protocol \cite{FHS11}\ also lead to the
following bounds on the strong locking capacity:%
\begin{equation}
Q\left(  \mathcal{N}\right)  \leq L_{S}\left(  \mathcal{N}\right)  \leq
L_{W}\left(  \mathcal{N}\right)  . \label{eq:strong-lock-bounds}%
\end{equation}
We first justify the bound $Q\left(  \mathcal{N}\right)  \leq L_{S}\left(
\mathcal{N}\right)  $, already observed in some sense in Ref.~\cite{F12}. The
strong locking capacity of the noiseless qubit channel is equal to one, due to
the existence of the FHS\ locking protocol (see
Example~\ref{ex:noiseless-qudit} below). By concatenating the FHS\ locking
protocol with a family of capacity-achieving quantum error correcting codes,
we obtain a family of strong locking protocols that achieve a strong locking
rate equal to the quantum capacity of the channel. The bound $L_{S}\left(
\mathcal{N}\right)  \leq L_{W}\left(  \mathcal{N}\right)  $ follows because a
strong locking protocol always meets the demands of a weak locking protocol
(recall Remark~\ref{rem:strong-implies-weak-lock}). The relationship between
the private capacity and the strong locking capacity is less clear. Indeed, a
private communication protocol for a quantum channel protects information only
from the environment of the channel (which we think of as the eavesdropper's
system). For this reason, it does not meet the demands of a strong locking
protocol. However, we could consider a \textquotedblleft strong
privacy\textquotedblright\ protocol in which the goal is to protect a message
from both the environment  and output 
of the channel, under the assumption that the party controlling these systems
does not have access to the shared key. In this case, the \textquotedblleft
strong private capacity\textquotedblright\ would always be equal to zero
because a sublinear amount of secret key is insufficient to get any
\textquotedblleft strong private capacity\textquotedblright\ out of the
channel. For this reason, the bounds in (\ref{eq:weak-lock-bounds}%
)-(\ref{eq:strong-lock-bounds}) are the best simple ones that we can derive
from operational considerations.

We can also consider the case in which a classical feedback channel is
available for free from the receiver to the sender. In this case, we denote
the resulting capacities with a superscript $^{\left(  \leftarrow\right)  }$.
By employing the same operational arguments as above, we find that the
following inequalities hold%
\begin{align}
Q^{\left(  \leftarrow\right)  }\left(  \mathcal{N}\right)   &  \leq P^{\left(
\leftarrow\right)  }\left(  \mathcal{N}\right)  \leq L_{W}^{\left(
\leftarrow\right)  }\left(  \mathcal{N}\right)  \leq C^{\left(  \leftarrow
\right)  }\left(  \mathcal{N}\right)  ,\\
Q^{\left(  \leftarrow\right)  }\left(  \mathcal{N}\right)   &  \leq
L_{S}^{\left(  \leftarrow\right)  }\left(  \mathcal{N}\right)  \leq
L_{W}^{\left(  \leftarrow\right)  }\left(  \mathcal{N}\right)  .
\label{eq:feedback-op-inequalities}%
\end{align}
Capacities assisted by classical feedback need not be equal to the unassisted
capacities. For example, it is known that the quantum and private capacities
assisted by classical feedback can be strictly larger than the corresponding
unassisted capacities \cite{LLS09}, and this is true even for the classical
capacity \cite{SS09}.
The locking capacity of quantum channels with classical feedback remains largely
an open question.

\subsection{Upper bounds on the locking capacity}

Let us define the information quantity $L_{W}^{\left(  u\right)  }\left(
\mathcal{N}\right)  $ as follows:%
\begin{equation}
L_{W}^{\left(  u\right)  }\left(  \mathcal{N}\right)  \equiv\max_{\left\{
p\left(  x\right)  ,\rho_{x}\right\}  } \left[I\left(  X;B\right)  -I_{\text{acc}%
}\left(  X;E\right) \right] ,\label{eq:weak-lock-upper-quantity}%
\end{equation}
where the above information quantities are evaluated with respect to a state
of the following form:%
\begin{equation}
\sum_{x}p_{X}\left(  x\right)  \left\vert x\right\rangle \left\langle
x\right\vert _{X}\otimes U_{A\rightarrow BE}^{\mathcal{N}}\left(  \rho
_{x}\right)  ,\label{eq:locking-cap-up-bnd-state}%
\end{equation}
$U_{A\rightarrow BE}^{\mathcal{N}}$ is an isometric extension of the channel
$\mathcal{N}$, and the superscript $\left(  u\right)  $ indicates that this
quantity will function as an upper bound on the locking capacity.
The following theorem establishes that the
regularization of $L_{W}^{\left(  u\right)  }\left(  \mathcal{N}\right)  $
provides an upper bound on the weak locking capacity of a quantum channel.
This bound is nontrivial given that the regularization of $L_{W}^{\left(
u\right)  }\left(  \mathcal{N}\right)  $ does not depend on the secret key
used in a given locking protocol.

\begin{theorem}
\label{thm:WLC}The weak locking capacity $L_{W}\left(  \mathcal{N}\right)  $
of a quantum channel $\mathcal{N}$ is upper bounded by the regularization of
$L_{W}^{\left(  u\right)  }\left(  \mathcal{N}\right)  $:%
\[
L_{W}\left(  \mathcal{N}\right)  \leq\lim_{n\rightarrow\infty}\frac{1}{n}%
L_{W}^{\left(  u\right)  }\left(  \mathcal{N}^{\otimes n}\right)  .
\]

\end{theorem}

\begin{proof}
The proof below places an upper bound on the weak locking capacity of a
quantum channel by considering the most general protocol for this task.
Suppose that the task is to generate shared, locked randomness rather than to
send a locked message (placing an upper bound on achievable rates for this
task gives an upper bound on achievable rates for the latter task, since a
protocol for the latter task can be used to accomplish the former task). The
most general protocol has Alice input her share of the key $K$ and her
variable $M$ into an encoder that outputs some systems $A^{n}$ to be fed into
the inputs of the channels. She then transmits these systems $A^{n}$ over the
channel, so that Bob receives the output systems $B^{n}$. Let the following
state describe all systems at this point in the protocol:%
\[
\omega_{MKB^{n}}\equiv\frac{1}{\left\vert \mathcal{M}\right\vert \left\vert
\mathcal{K}\right\vert }\sum_{m,k}\left\vert m\right\rangle \left\langle
m\right\vert _{M}\otimes\left\vert k\right\rangle \left\langle k\right\vert
_{K}\otimes\mathcal{N}_{A\rightarrow B}^{\otimes n}\left(  \rho_{k,m}\right)
.
\]
Bob inputs his share of the key $K$ and the systems $B^{n}$ into a decoder
$\mathcal{D}_{KB^{n}\rightarrow\hat{M}}$\ to recover $\hat{M}$, which is his
estimate of Alice's variable $M$. The final state of the protocol is given by%
\[
\omega_{M\hat{M}}^{\prime}\equiv\frac{1}{\left\vert \mathcal{M}\right\vert
\left\vert \mathcal{K}\right\vert }\sum_{m,k}\left\vert m\right\rangle
\left\langle m\right\vert _{M}\otimes\mathcal{D}_{KB^{n}\rightarrow\hat{M}%
}\left[  \left\vert k\right\rangle \left\langle k\right\vert _{K}%
\otimes\mathcal{N}_{A\rightarrow B}^{\otimes n}\left(  \rho_{k,m}\right)
\right]  .
\]
If the protocol is any good for locking the message $M$, then the ideal
distribution of $M$ and $\hat{M}$ deviates from the actual distribution of
these variables by no more than $\varepsilon$, in the sense that%
\[
\Big\Vert \overline{\Phi}_{M\hat{M}}-\omega_{M\hat{M}}^{\prime}\Big\Vert_{1}\leq\varepsilon,
\]
where%
\[
\overline{\Phi}_{M\hat{M}}\equiv\frac{1}{\left\vert \mathcal{M}\right\vert
}\sum_{m}\left\vert m\right\rangle \left\langle m\right\vert _{M}%
\otimes\left\vert m\right\rangle \left\langle m\right\vert _{\hat{M}}.
\]
The above condition is equivalent to the condition that $\Pr\{\hat{M}\neq
M\}\leq\varepsilon/2$ because
$$\frac{1}{2}\Big\Vert \overline{\Phi}_{M\hat
{M}}-\omega_{M\hat{M}}^{\prime}\Big\Vert _{1}=\Pr\{\hat{M}\neq M\} .$$ Also,
from Remark~\ref{rem:FA-acc-bound}, Eve's accessible information $I_{\text{acc}}\left(
M;E^{n}\right)  $ about the variable $M$ is bounded from above by
$\varepsilon''n$, where
$\varepsilon''\equiv h_2(\varepsilon/2)/n + \varepsilon R/2$, whenever
(\ref{eq:sec-crit-var-dist}) is satisfied.
We can now proceed with bounding achievable rates for any
locking protocol:%
\begin{align*}
nR &  =H\left(  M\right)  _{\overline{\Phi}}\\
&  =I(M;\hat{M})_{\overline{\Phi}}\\
&  \leq I(M;\hat{M})_{\omega^{\prime}}+n\varepsilon^{\prime}\\
&  \leq I\left(  M;B^{n}K\right)  _{\omega}+n\varepsilon^{\prime}\\
&  =I\left(  M;B^{n}\right)  _{\omega}+I\left(  M;K|B^{n}\right)  _{\omega
}+n\varepsilon^{\prime}\\
&  \leq I\left(  M;B^{n}\right)  _{\omega}-I_{\text{acc}}\left(
M;E^{n}\right)  _{\omega}+o\left(  n\right)  +n\varepsilon^{\prime
}+n\varepsilon''\\
&  \leq L_{W}^{(u)}\left(  \mathcal{N}^{\otimes n}\right)  +o\left(  n\right)
+n\varepsilon^{\prime}+n\varepsilon''.
\end{align*}
The first equality follows from the assumption that the random variable $M$ is
a uniform random variable. The second equality is an identity because
$H(M|\hat{M})=0$ for the ideal distribution on $M$ and $\hat{M}$. The first
inequality follows from an application of the Alicki-Fannes-Audenart
inequality (continuity of entropy) \cite{AF04,A07}, where $\varepsilon
^{\prime}$ is a function of $\varepsilon$ that approaches zero as
$\varepsilon\rightarrow0$. The second inequality follows from an application
of quantum data processing (both $B^{n}$ and $K$ are fed into the decoder to
produce $\hat{M}$). The third equality follows from an application of the
chain rule for mutual information. The third inequality follows from the upper
bound
\[
I\left(  M;K|B^{n}\right)  \leq H\left(  K|B^{n}\right)  \leq H\left(
K\right)  \leq o\left(  n\right)  ,
\]
(the assumption that the secret key rate is sublinear) and from the accessible
information bound $I_{\text{acc}}\left(  M;E^{n}\right)  \leq\varepsilon$. The final
inequality follows from optimizing over all distributions, so that we have%
\[
R\leq\lim_{n\rightarrow\infty}\frac{1}{n}L_{W}^{\left(  u\right)  }\left(
\mathcal{N}^{\otimes n}\right)  .
\]
in the limit as $n$ becomes large and as $\varepsilon\rightarrow0$.
\end{proof}

\begin{theorem}
\label{thm:up-bnd-strong-locking-cap}The strong locking capacity $L_{S}\left(
\mathcal{N}\right)  $ of a quantum channel $\mathcal{N}$ is upper bounded as%
\[
L_{S}\left(  \mathcal{N}\right)  \leq\lim_{n\rightarrow\infty}\frac{1}{n}%
L_{S}^{\left(  u\right)  }\left(  \mathcal{N}^{\otimes n}\right)  ,
\]
where%
\[
L_{S}^{\left(  u\right)  }\left(  \mathcal{N}\right)  \equiv\max_{\left\{
p\left(  x\right)  ,\rho_{x}\right\}  } \left[I\left(  X;B\right)
-I_{\operatorname{acc}}\left(  X;BE\right) \right] ,
\]
and the information quantities are with respect to the state in
(\ref{eq:locking-cap-up-bnd-state}).
\end{theorem}

\begin{proof}
The proof of this theorem is nearly identical to the proof of the one above.
However, we employ the bound on the
accessible information $I_{\text{acc}}\left(  M;A^{n}\right)
=I_{\text{acc}}\left(  M;B^{n}E^{n}\right)$ from
Definition~\ref{def:strong-lock-protocol} and Remark~\ref{rem:FA-acc-bound} instead.
\end{proof}

\begin{remark}
Observe that the bounds in the above theorem hold even if the key is allowed
to be a sublinear size quantum system, as in the locking schemes discussed in
\cite{DFHL10}.
\end{remark}

It is an interesting and important open question to determine if the upper
bounds given in the above theorems are achievable.

\subsubsection{Entanglement-breaking channels have zero locking capacity}

The above theorems and a further analysis allow us to determine that both the
strong and weak locking capacities of an entanglement-breaking channel are
equal to zero.

\begin{definition}
[Entanglement-breaking channel \cite{HSR03}]A channel $\mathcal{N}%
_{\operatorname{EB}}$ is entanglement-breaking if the output state is
separable whenever it acts on one share of an entangled state:%
\[
\left(  \operatorname{id}_{R}\otimes\mathcal{N}_{\operatorname{EB}}\right)
\left(  \rho_{RA}\right)  =\sum_{x}p_{X}\left(  x\right)  \sigma_{R}%
^{x}\otimes\omega_{B}^{x},
\]
where $p_{X}\left(  x\right)  $ is a probability distribution, each
$\sigma_{R}^{x}$ is a state on the reference system $R$, and each $\omega
_{B}^{x}$ is a state on the channel output system $B$.
\end{definition}

\begin{theorem}
Both the strong and weak locking capacities of an entanglement-breaking
channel $\mathcal{N}_{\operatorname{EB}}$ are equal to zero:%
\[
L_{W}(\mathcal{N}_{\operatorname{EB}})=L_{S}(\mathcal{N}_{\operatorname{EB}%
})=0.
\]
\label{thm:ent-break}
\end{theorem}

\begin{proof}
The proof of this theorem exploits the upper bound derived in
Theorem~\ref{thm:WLC} and the fact that $L_{W}(\mathcal{N}_{\text{EB}})\geq
L_{S}(\mathcal{N}_{\text{EB}})$. We know from Ref.~\cite{HSR03} that any
entanglement-breaking channel has a representation with rank-one Kraus
operators, so that its action on an input density operator is given by%
\[
\mathcal{N}_{\text{EB}}\left(  \rho\right)  =\sum_{y}\left\vert \phi
_{y}\right\rangle _{B}\left\langle \psi_{y}\right\vert _{A}\rho\left\vert
\psi_{y}\right\rangle _{A}\left\langle \phi_{y}\right\vert _{B},
\]
for some set of vectors $\left\{  \left\vert \psi_{y}\right\rangle
_{A}\right\}  $ such that $\sum_{y}\left\vert \psi_{y}\right\rangle
\left\langle \psi_{y}\right\vert _{A}=I_A$ and a set of states $\left\{
\left\vert \phi_{y}\right\rangle _{B}\right\}  $. An isometric extension of
the channel is then given by%
\[
U_{A\rightarrow BE}^{\mathcal{N}_{\text{EB}}}\equiv\sum_{y}\left\vert \phi
_{y}\right\rangle _{B}\left\langle \psi_{y}\right\vert _{A}\otimes\left\vert
y\right\rangle _{E},
\]
with $\left\{  \left\vert y\right\rangle _{E}\right\}  $ an orthonormal basis
for the environment. From this representation, it is clear that the channel to
the environment is of the form:%
\[
\mathcal{N}_{\text{EB}}^{c}\left(  \rho\right)  =\sum_{y,z}\left\langle
\psi_{y}\right\vert \rho\left\vert \psi_{z}\right\rangle _{A} \ \left\langle
\phi_{z}|\phi_{y}\right\rangle _{B}\ \left\vert y\right\rangle \left\langle
z\right\vert _{E},
\]
and the environment can simulate the channel to the receiver by first
performing a von Neumann measurement in the basis $\left\{  \left\vert
y\right\rangle \right\}  $ followed by a preparation of the state $\left\vert
\phi_{y}\right\rangle _{B}$ conditioned on the measurement outcome being $y$.

Now consider the information quantity $L_{W}^{\left(  u\right)  }%
(\mathcal{N}_{\text{EB}})$ defined in (\ref{eq:weak-lock-upper-quantity}).
Theorem~\ref{thm:WLC} states that the regularization of this quantity is an
upper bound on the weak locking capacity. For any finite $n$, we can always
pick the measurement to be a tensor-product von Neumann measurement of the
form mentioned above, giving that%
\[
I_{\text{acc}}\left(  X;E^{n}\right)  \geq I\left(  X;Y^{n}\right)  ,
\]
where $Y^{n}$ is the random variable corresponding to the measurement
outcomes. Due to the structural relationship given above (the fact that the
environment can simulate the channel to the receiver by preparing $n$ quantum
states $\left\vert \phi_{y_{1}}\right\rangle \otimes\cdots\otimes\left\vert
\phi_{y_{n}}\right\rangle $ from the measurement outcomes $y^{n}$), we find
that%
\[
I\left(  X;Y^{n}\right)  \geq I\left(  X;B^{n}\right)  ,
\]
by an application of the quantum data processing inequality. This is
equivalent to $I\left(  X;B^{n}\right)  -I\left(  X;Y^{n}\right)  \leq0$,
which implies that $\lim_{n\rightarrow\infty}\frac{1}{n}L_{W}^{\left(
u\right)  }\left(  \mathcal{N}_{\text{EB}}^{\otimes n}\right)  = 0$ and thus
that the weak locking capacity vanishes for any entanglement-breaking channel.
\end{proof}

\begin{remark}
The importance of the above theorem is the conclusion that a channel should be
able to preserve entanglement between a purification of the channel input and
its output in order for it to be able to lock information. If it is not able
to (i.e., if it is entanglement-breaking), then the locking capacity is equal
to zero. Ref.~\cite{Boixo} suggested that entanglement does not play a role in quantum data locking,
but this theorem shows that it does in any realistic implementation of a locking protocol.
\end{remark}

\begin{remark}
It should be possible to provide a rigorous generalization of 
this result to entanglement-breaking channels
defined over general infinite-dimensional spaces using the techniques from Ref.~\cite{H08}.
For example, it is known that a lossy bosonic channel becomes entanglement-breaking 
when the environment injects a thermal state with sufficiently high photon number \cite{H08}.
However, we leave this question open for future work.
\end{remark}

\subsubsection{Protocols with classical simulations have zero strong locking
rate}

It is important to determine the conditions for when the locking rate of
a given protocol is 
zero, so that we can distinguish between the classical
and quantum regimes for locking. In this regard, we can exclude all protocols
that have a classical simulation in the following sense:

\begin{definition}
[Classical simulation]We say that a locking protocol has a classical
simulation if the receiver's decoding consists of performing a measurement on
the output of the channel that is independent of the key $K$, followed by a
classical post-processing of the measurement output and the key to produce an
estimate of the transmitted message.
\end{definition}

\begin{theorem}
\label{thm:classical-sim}The strong locking rate of 
any locking protocol with a classical simulation is equal to zero.
\end{theorem}

\begin{proof}
The fact that this theorem should hold might be obvious, but nevertheless we
provide a proof. The setup for this proof is similar to that in the proof of
Theorems~\ref{thm:WLC} and \ref{thm:up-bnd-strong-locking-cap}, with the
exception that the decoder first performs a key-independent measurement of the
channel output to produce a random variable $Y$. The decoder then processes
the random variable $Y$ and the key $K$ to produce an estimate $\hat{M}$ of
the sender's message. We can bound the rate $R$ of this protocol as follows:%
\begin{align*}
nR &  =H\left(  M\right)  _{\overline{\Phi}}\\
&  =I(M;\hat{M})_{\overline{\Phi}}\\
&  \leq I(M;\hat{M})_{\omega^{\prime}}+n\varepsilon^{\prime}\\
&  \leq I\left(  M;YK\right)  _{\omega}+n\varepsilon^{\prime}\\
&  =I\left(  M;Y\right)  _{\omega}+I\left(  M;K|Y\right)  _{\omega
}+n\varepsilon^{\prime}\\
&  \leq I\left(  M;Y\right)  _{\omega}-I_{\text{acc}}\left(  M;B^{n}%
E^{n}\right)  _{\omega}+o\left(  n\right)  +n\varepsilon^{\prime}%
+n\varepsilon''\\
&  \leq I\left(  M;Y\right)  _{\omega}-I\left(  M;Y\right)  _{\omega}+o\left(
n\right)  +n\varepsilon^{\prime}+n\varepsilon''\\
&  =o\left(  n\right)  +n\varepsilon^{\prime}+n\varepsilon''.
\end{align*}
The first three lines above are exactly the same as those in the proof of
Theorem~\ref{thm:WLC}. The second inequality follows from quantum data
processing. The third equality is the chain rule. The third inequality follows
from the condition $I_{\text{acc}}\left(  M;B^{n}E^{n}\right)  _{\omega}%
\leq\varepsilon''n$,
 with
$\varepsilon''\equiv h_2(\varepsilon/2)/n + \varepsilon R/2$, whenever
(\ref{eq:sec-crit-var-dist}) is satisfied,
which should hold for any strong locking protocol. Also, it
follows because $I\left(  M;K|Y\right)  _{\omega}\leq H\left(  K\right)  \leq
o\left(  n\right)  $. Finally, the adversary can choose her processing of the
$B^{n}E^{n}$ systems to be a discarding of $E^{n}$ followed by whatever
key-independent measurement of $B^{n}$ that the receiver is performing to
produce $Y$. Thus, it follows that $I_{\text{acc}}\left(  M;B^{n}E^{n}\right)
_{\omega}\geq I(M;Y)$. The statement that the strong locking rate is equal
to zero follows by taking the limit as $n\rightarrow\infty$ and $\varepsilon
\rightarrow0$.
\end{proof}

As a corollary of the above theorem, we find the following:

\begin{corollary}
If a protocol does not consume any secret key at all, then the strong locking rate is equal to zero.
\end{corollary}
\begin{proof}
This follows simply because the receiver's measurement on the channel outputs 
does not depend on a key for a scheme that does not use any key at all.
\end{proof}

\subsubsection{The private and quantum capacity are equal to the weak locking capacity for
particular Hadamard channels}

In this section, we prove that if the channel is such that the map
from the input to the environment
is a quantum-to-classical channel, i.e., of the following form:
\begin{equation}
\rho \to \sum_x \text{Tr}\{A_x \rho A_x^\dag\} \, \vert x \rangle \langle x \vert,
\label{eq:q-c-env}
\end{equation}
for some orthonormal basis $\{ \vert x \rangle \}$ and
where $\sum_x A_x^\dag A_x = I$, then the weak locking capacity of such a channel
is equal to its private and quantum capacity. This result follows simply because the systems
received by the environment are already classical, so that the best measurement
for the adversary to perform is given by $\{\vert x \rangle \langle x \vert \}$
on each channel use. Any measurement other
than this one will have
a mutual information with the message lower than
this measurement's mutual information by a simple data processing argument. Furthermore, since
the systems given to the environment are classical, the Holevo information of the environment
 with the input is
equal to the accessible information of the environment with the input for such channels. 

For such channels, the map from the
input to the output is of the following form:
\begin{equation}
\rho \to \sum_x A_x \rho A_x^\dag \otimes \vert x \rangle \langle x \vert
\label{eq:Hadamard-qc-env}
\end{equation}
because the operator $\sum_x A_x (\cdot) \otimes \vert x \rangle \otimes \vert x \rangle$
is an isometric extension of the channel in (\ref{eq:q-c-env}).
A notable example of such a channel is the ``photon detected-jump'' channel, described
in Ref.~\cite{GJWZ10}. Channels of the form in (\ref{eq:Hadamard-qc-env})
 are examples of Hadamard channels,
which are generally defined as channels complementary to entanglement-breaking
ones \cite{K03had,KMNR07}.

We state the above result as the following theorem:
\begin{theorem}
\label{thm:Hadamard-qc-env}The weak locking capacity of a channel of the
form in (\ref{eq:Hadamard-qc-env}) is equal to its private and quantum capacity and is given by the following expression:
$$
\max_{\left\{  p_{X}\left(
x\right)  ,\rho_{x}\right\}  }\left[  I\left(  X;B\right)  -I\left(
X;E\right)  \right], 
$$
where the information quantities are evaluated with respect to the following state:
$$
\sum_x p_X(x) \vert x \rangle \langle x \vert_X \otimes U^{\mathcal{N}}_{A\to BE}(\rho_x)
$$
with $U^{\mathcal{N}}_{A\to BE}$ an isometric extension of the channel $\mathcal{N}$.
\end{theorem}

\begin{proof}
A proof of this theorem follows the intuition mentioned above. In particular, we know from
Refs.~\cite{ieee2005dev,1050633} that the following formula is equal to the private capacity of any channel:
$$
P\left(  \mathcal{N}\right)
=\lim_{n\rightarrow\infty}\frac{1}{n}\left[  \max_{\left\{  p_{X}\left(
x\right)  ,\rho^{(n)}_{x}\right\}  }\left[  I\left(  X;B^{n}\right)  -I\left(
X;E^{n}\right)  \right]  \right]  
$$
Now, since we are assuming the channel to
the environment to have the form given in (\ref{eq:q-c-env}),
the systems given to the environment are classical so that the accessible information
$I_\text{acc}(X;E^n)$ is equal to the Holevo information  $I\left(
X;E^{n}\right)$ for any finite $n$. Thus, our upper bound from Theorem~\ref{thm:WLC}
on the weak locking capacity
of such a channel
is equal to the expression given above for its private capacity. Furthermore, all Hadamard channels
are degradable \cite{BHTW10}, meaning that the receiver can simulate the map from the input to the
environment by acting with a degrading map on his system. Finally, it is known that the
expression for the private capacity ``single-letterizes'' to the form in the statement of the theorem
for degradable channels and that the quantum capacity is equal to the private capacity for
such channels \cite{S08}.
\end{proof}

\begin{remark}
Theorem~\ref{thm:Hadamard-qc-env}
demonstrates that it suffices to use a private capacity achieving code
for channels of the form in (\ref{eq:Hadamard-qc-env}), with the benefit that
these private communication
codes do not require
the consumption of any secret key.
That is, there is no need to devise an exotic information locking protocol for such
channels in order to achieve their weak locking capacity.
\end{remark}

\subsubsection{Quantum discord-based upper bound on the gap between weak
locking capacity and private capacity}

The quantum discord is an asymmetric measure that quantifies the quantum
correlation in a bipartite quantum state \cite{HZ01}. For a given bipartite
quantum state $\rho_{AB}$, the quantum mutual information $I\left(
A;B\right)  _{\rho}$ quantifies all of the bipartite correlations in
$\rho_{AB}$, while $\max_{\Lambda_{A\rightarrow X}}I\left(  X;B\right)  $ is
meant to capture the classical correlations in the state that are recoverable
by performing a local measurement on the $A$ system \cite{HV01}. Thus, the
idea behind the quantum discord $D\left(  A,B\right)  _{\rho}$\ is to quantify
the quantum correlations in a state by subtracting out the classical
correlation from the total correlation:%
\[
D\left(  A,B\right)  _{\rho}\equiv I\left(  A;B\right)  _{\rho}-\max
_{\Lambda_{A\rightarrow X}}I\left(  X;B\right)  .
\]
Ollivier and Zurek originally described the quantum discord as the correlations lost during
a measurement process \cite{HZ01}.

Our upper bound on the weak locking capacity from\ Theorem~\ref{thm:WLC}%
\ appears similar to the above formula for quantum discord. Indeed, we can
place an upper bound on the gap between the weak locking capacity and the
private capacity of a quantum channel in terms of the discord between the
environment of the channel and the classical variable sent into the channel.
We can also interpret this merely as the gap between the Holevo information of
the environment and its accessible information. It is clear why this gap is
related to quantum discord. In a private communication protocol, the security guarantee
is with respect to the Holevo information, while in a locking protocol, the guarantee
is with respect to the accessible information. Thus, the gap between the two capacities
should be related to the correlations lost during Eve's measurement.

\begin{proposition}
The gap between the weak locking capacity and the private capacity of a
quantum channel is no larger than%
\begin{align*}
L_{W}\left(  \mathcal{N}\right)  -P\left(  \mathcal{N}\right) 
& \leq
\lim_{n\rightarrow\infty}\frac{1}{n}\left[  \max_{\left\{  p_{X}\left(
x\right)  ,\rho_{x}\right\}  }I\left(  X;E^{n}\right)  -I_{\operatorname{acc}%
}\left(  X;E^{n}\right)  \right]   \\
& = \lim_{n\rightarrow\infty}\frac{1}{n}\left[  \max_{\left\{  p_{X}\left(
x\right)  ,\rho_{x}\right\}  } D(E^n,X) \right] ,
\end{align*}
where the entropies for any finite $n$ are with respect to a state of the
following form:%
\[
\sum_{x}p_{X}\left(  x\right)  \left\vert x\right\rangle \left\langle
x\right\vert _{X}\otimes\mathcal{N}_{A\rightarrow E}^{\otimes n}\left(
\rho_{x}\right)  ,
\]
and $\mathcal{N}_{A\rightarrow E}$ is the channel complementary to
$\mathcal{N}_{A\rightarrow B}=\mathcal{N}$.
\end{proposition}

\begin{proof}
Consider that for any finite $n$, we have the bound%
\begin{align*}
I\left(  X;B^{n}\right)  -I_{\text{acc}}\left(  X;E^{n}\right)   &  =I\left(
X;B^{n}\right)  -I\left(  X;E^{n}\right)  +I\left(  X;E^{n}\right)
-I_{\text{acc}}\left(  X;E^{n}\right) \\
&  \leq\max_{\left\{  p_{X}\left(  x\right)  ,\rho_{x}\right\}  }\left[
I\left(  X;B^{n}\right)  -I\left(  X;E^{n}\right)  \right]  +\max_{\left\{
p_{X}\left(  x\right)  ,\rho_{x}\right\}  }\left[  I\left(  X;E^{n}\right)
-I_{\text{acc}}\left(  X;E^{n}\right)  \right]  .
\end{align*}
Then by using the bound from Theorem~\ref{thm:WLC}, the inequality above, and
the characterization of the private capacity as $P\left(  \mathcal{N}\right)
=\lim_{n\rightarrow\infty}\frac{1}{n}\left[  \max_{\left\{  p_{X}\left(
x\right)  ,\rho_{x}\right\}  }\left[  I\left(  X;B^{n}\right)  -I\left(
X;E^{n}\right)  \right]  \right]  $, the bound in the statement of the theorem follows.
\end{proof}

\subsection{Examples}

\begin{example}
[Noiseless qudit channel]\label{ex:noiseless-qudit}The noiseless qudit channel
trivially has weak locking capacity equal to $\log_{2}d$, where $d$ is the
dimension of the input and output for the channel. The reason for this is that
an isometric extension of this channel has the following form:%
\[
\sum_{i}\left\vert i\right\rangle _{B}\left\langle i\right\vert _{A}%
\otimes\left\vert \phi\right\rangle _{E}.
\]
In this case, Eve's state is independent of the input, so that her accessible
information is always equal to zero (even without coding in any way).\newline%
\newline However, the noiseless qudit channel nontrivially has strong
locking capacity also equal to $\log_{2}d$. This follows from the results of Fawzi
\textit{et al}.~\cite{FHS11}, in which they demonstrated the existence of a
locking protocol that locks $n$ dits using $4  \log_{2}\left(
1/\varepsilon\right)  +O\left(  \log_{2}\log_{2}\left(
1/\varepsilon\right)  \right)  $ bits of key while having the variational distance
in (\ref{eq:sec-crit-var-dist}) for any eavesdropper measurement
no larger than $\varepsilon$, for
an eavesdropper who obtains the full output of the
noiseless channel. 
Thus, this scheme is an $\left(  n,\log_{2}d,\varepsilon\right)  $ locking
protocol that consumes secret key at a rate equal to%
\[
\frac{1}{n}\left[  4  \log_{2}\left(  1/\varepsilon\right)  
+O\left(  \log_{2}\log_{2}\left(  1/\varepsilon\right)  \right)  \right]  .
\]
So, for any fixed $\varepsilon>0$, we can take $n$ large so that the secret
key rate vanishes in this limit, while the eavesdropper will not be able to do much
better than to randomly guess the message.
Thus, this construction gives a scheme to achieve the rate
$\log_{2}d$. Since the strong locking capacity of the noiseless qudit channel
cannot be any larger than $\log_{2}d$, this proves that it is equal to
$\log_{2}d$ for this channel.
\end{example}

In reality, one does not ever have access to perfectly independent uses of a quantum
channel, as this is just an idealization. As such,
it can be helpful to define the ``one-shot'' locking capacity for a
single use of a quantum channel. We provide such a definition below:

\begin{definition}
[One-shot locking capacity]\label{def:one-shot-locking-cap}The $\varepsilon
$-one-shot locking capacity of a quantum channel is the maximum number of
locked bits that a sender can transmit to a receiver such that the receiver
can recover the message with average error probability less than
$\varepsilon>0$ and such that the total variational distance\ of the message
distribution conditioned on the eavesdropper's measurement outcome$~x$ with
the unconditioned message distribution $p_{M}$\ is no larger than
$\varepsilon$:
\[
  \sum_m \vert p_{M|X}(m|x) - p_{M}(m)\vert  \leq\varepsilon.
\]
We also demand that the number of secret key bits used is $O\left(  \log
_{2}\log_{2}\left\vert \mathcal{M}\right\vert \right)  $. Similar to the
IID\ case, we can distinguish between weak and strong locking capacities.
\end{definition}

\begin{example}
[Depolarizing channel]Recall that the quantum depolarizing channel is defined
as%
\[
\rho\rightarrow\left(  1-p\right)  \rho+p\frac{I}{d},
\]
where $p\in\left[  0,1\right]  $ characterizes the noisiness of the channel
and $d$ is its dimension. For sufficiently large $d$, the $\varepsilon
$-one-shot strong locking capacity of the depolarizing channel is equal to its
$\varepsilon$-one-shot classical capacity (defined similarly as above---see
Ref.~\cite{WR12}, for example). This result follows simply because any unitary
encoding commutes with the action of the depolarizing channel on the input
state, and we can employ the FHS\ protocol combined with an $\varepsilon
$-one-shot classical capacity achieving code, in order to achieve the same
$\varepsilon$-one-shot strong locking capacity of the depolarizing channel.
\newline\newline
While easy to prove, this example illustrates the subtle interplay between locking, entanglement and classical communication. The fact that the depolarizing channel's one-shot strong locking and classical capacities match regardless of the strength of the noise would seem to leave little room for quantum correlations to play any role. Indeed, it seems hard to square this result with Theorem~\ref{thm:ent-break}'s statement that entanglement-breaking channels have zero strong locking capacity, which is easily adapted to the one-shot setting. The resolution is that for any fixed but arbitrarily large amount of noise $p$, the depolarizing channel eventually ceases to be entanglement breaking for some sufficiently large $d = \text{poly}(1/p)$~\cite{gurvits2002largest}.
\newline%
\newline Our best known characterization of the locking capacity of the
IID\ memoryless depolarizing channel is in terms of the operational
inequalities given in (\ref{eq:weak-lock-bounds})-(\ref{eq:strong-lock-bounds}).
\end{example}

\begin{example}
[Erasure channel]Consider a $d$-dimensional quantum erasure channel defined as%
\[
\rho\rightarrow\left(  1-p\right)  \rho+p\left\vert e\right\rangle
\left\langle e\right\vert ,
\]
where $\left\vert e\right\rangle $ is an erasure flag state that is orthogonal
to the $d$-dimensional input state. For this channel, a unitary acting on the
input commutes with the action of the channel, so that the same argument as
above demonstrates that the $\varepsilon$-one-shot strong locking capacity of this
channel is equal to its $\varepsilon$-one-shot classical capacity for
sufficiently large $d$.\newline\newline The feedback-assisted weak and strong
locking capacities of the memoryless erasure channel are at least $\left(
1-p\right)  ^{2}$ for $p\leq1/2$ and $\left(  1-p\right)  /\left(
1+2p\right)  $ for $p\geq1/2$. Furthemore, they are no larger than $1-p$.
These results follow from the best known lower bounds on the quantum capacity
of the erasure channel assisted by classical feedback \cite{LLS09}, the fact
that the feedback-assisted classical capacity of the erasure channel cannot
exceed $1-p$, and the operational inequalities in
(\ref{eq:feedback-op-inequalities}).
\end{example}

\begin{example}
[Parallelized locking protocols]A simple parallelized protocol (as mentioned
in Ref.~\cite{QEM}) is to employ the FHS\ protocol for each use of a
memoryless depolarizing or erasure channel. However, the best known statement
regarding the parallel composition of locking protocols is given by
Proposition~2.4 of Ref.~\cite{FHS11}. That is, if one locking protocol
guarantees that the total variational distance\ of a message distribution
conditioned on the eavesdropper's measurement outcome$~x_{1}$ with the
unconditioned message distribution $p_{M}$\ is no larger than $\varepsilon
_{1}$:%
\[
\sum_{m_1} \vert p_{M_{1}|X_{1}}(m_1|x_1) - p_{M_{1}}(m_1) \vert  \leq\varepsilon_{1},
\]
and another guarantees it is no larger than $\varepsilon_{2}$:%
\[
\sum_{m_2} \vert p_{M_{2}|X_{2}}(m_2|x_2) - p_{M_{2}}(m_2) \vert  \leq\varepsilon_{2},
\]
then the parallel composition of these protocols guarantees a total
variational distance no larger than $\varepsilon_{1}+\varepsilon_{2}$:%
\[
\sum_{m_1,m_2} \vert  p_{  M_{1},M_{2}  |X}(m_1,m_2|x) - p_{M_{1}%
,M_{2}  }(m_1,m_2) \vert  \leq\varepsilon_{1}+\varepsilon_{2}.
\]
Then consider a simple parallelized
protocol consisting of $n$ uses of a $d$-dimensional channel, where we suppose
that each channel use has a guarantee that the variational distance (as above)
is no larger than $\gamma>0$. Parallel composition of the locking protocols
guarantees that the variational distance for the $n$ channel uses is no larger
than $\gamma n$. By applying the Fannes-Audenaert inequality
\cite{Fannes73,A07} as in Proposition~3.2 of Ref.~\cite{FHS11}, one finds the
following bound on the accessible information of the adversary:%
\[
\left(  \gamma n\right)  \log d_{E}^{n}+h_{2}\left(  \gamma n\right)  ,
\]
where $d_{E}$ is the dimension of the environment for a single channel use.
Thus, the number of secret key bits needed to guarantee that Eve's accessible
information is no larger than $n \varepsilon$ is equal to $O\left(  n\log
_{2}\left(  1/\varepsilon\right)  \right)  $, so that the rate of key used in
this scheme grows linearly with the number of channel uses. Clearly,
this approach is less desirable than simply using a one-time pad combined with
a classical capacity achieving code. For this latter protocol, the rate of key
is a fixed constant independent of the number of channel uses and the protocol
guarantees perfect secrecy from an adversary with access to a quantum memory.
\newline\newline
In information theory, results for memoryless channels usually follow straightforwardly from their one-shot counterparts. The linear key growth incurred when parallelizing locking protocols prevents us from quickly concluding that the non-one-shot strong locking capacities of the depolarizing and erasure channels match their classical capacities. Moreover, the covariance argument used to draw that conclusion does not translate directly to the setting of many channel uses. We therefore leave it as an open question to determine whether the equivalence persists beyond the one-shot setting.
\end{example}

\section{Upper bounds on the locking capacity when restricting to coherent-state encodings}

\label{CSQEM}In this section, we prove that there are fundamental limitations
on the locking capacity of channels when we restrict ourselves to coherent-state
encodings. In particular, we prove that the strong locking
capacity of any quantum channel cannot be any larger than%
\[
g\left(  N_{S}\right)  -\log_{2}\left(  1+N_{S}\right)  ,
\]
where $g(x) \equiv (x+1) \log_2(x+1) - x \log_2 x$,
when restricting to coherent-state encodings with mean input photon number
$N_{S}$. Observe that $g\left(  N_{S}\right)  -\log_{2}\left(  1+N_{S}\right)
\leq\log_{2}\left(  e\right)  $, and this latter bound is independent of the
photon number used for the coherent-state codewords. An intuitive (yet not
fully rigorous) reason for why we obtain this bound is that $\log_{2}\left(
1+N_{S}\right)  $ is the rate of information that an adversary
can recover about the message simply by performing
heterodyne detection on each input to the channel, while $g\left(
N_{S}\right)  $ is an upper bound on the classical capacity of any channel
with mean input photon number $N_{S}$. Thus, the difference of these two
quantities should be a bound on the strong locking capacity.

We also prove that the weak locking capacity of a pure-loss bosonic channel
cannot be any larger than the sum of its private capacity and%
\[
g\left(  \left(  1-\eta\right)  N_{S}\right)  -\log_{2}\left( 
    1+\left(
1-\eta\right)N_{S}  \right)  ,
\]
when restricting to coherent-state encodings with mean photon number $N_{S}$,
where $\eta\in\left[  0,1\right]  $ is the transmissivity of the channel. As before,
$g\left(  \left(  1-\eta\right)  N_{S}\right)  -\log_{2}\left( 
    1+\left(
1-\eta\right)N_{S}  \right)  \leq
\log_{2}\left(  e\right)  $, which is independent of the photon number.

We consider a coherent-state locking protocol in which the encrypted states
$\{U_{k}|m\rangle\}$ are generalized to a set of $n$-mode coherent states
$\{|\alpha^{n}\left(  m,k\right)  \rangle\}_{m\in\mathcal{M},k\in\mathcal{K}}%
$, where $|\alpha^{n}\left(  m,k\right)  \rangle$ is an $n$-fold tensor
product of coherent states:%
\[
\left\vert \alpha^{n}\left(  m,k\right)  \right\rangle \equiv\left\vert
\alpha_{1}\left(  m,k\right)  \right\rangle \otimes\cdots\otimes\left\vert
\alpha_{n}\left(  m,k\right)  \right\rangle .
\]

\begin{definition}
[Coherent-state locking protocol]\label{def:coh-state-QEM}A coherent-state
locking protocol\ consists of coherent-state codewords $\left\{  |\alpha
^{n}\left(  m,k\right)  \rangle\right\}  _{m\in\mathcal{M},k\in\mathcal{K}}$
depending upon the message $m$ and the key value$~k$. These codewords are then
transmitted over a quantum channel to be decoded by a receiver.
\end{definition}

\begin{theorem}
The strong locking capacity of any channel when restricting to coherent-state
encodings with mean photon number $N_{S}$ is upper bounded by $g\left(
N_{S}\right)  -\log_{2}\left(  1+N_{S}\right)  $.
\end{theorem}

\begin{proof}
As described in Definition~\ref{def:coh-state-QEM}, the encoder for such a
scheme prepares a coherent-state codeword $|\alpha^{n}\left(  m,k\right)
\rangle$\ at the input of $n$ uses of a quantum channel $\mathcal{N}$,
depending upon the message~$m$ and the key value $k$. It is useful for us to
consider the following classical-quantum state, which describes the state of
the message, key, and input to many uses of the channel:%
\begin{equation}
\rho_{MKA^{n}}=\frac{1}{|\mathcal{M}||\mathcal{K}|}\sum_{m,k}|m,k\rangle
\langle m,k|_{MK}\otimes|\alpha^{n}\left(  m,k\right)  \rangle\langle
\alpha^{n}\left(  m,k\right)  |_{A^{n}}\,.
\end{equation}
The state after the isometric extension of the channel (unique up to unitaries
acting on the environment)\ acts is then as follows:%
\[
\rho_{MKB^{n}E^{n}}=\frac{1}{|\mathcal{M}||\mathcal{K}|}\sum_{m,k}%
|m,k\rangle\langle m,k|_{MK}\otimes U_{A^{n}\rightarrow B^{n}E^{n}%
}^{\mathcal{N}}\left(  |\alpha^{n}\left(  m,k\right)  \rangle\langle\alpha
^{n}\left(  m,k\right)  |_{A^{n}}\right)  ,
\]
where $U_{A^{n}\rightarrow B^{n}E^{n}}^{\mathcal{N}}$ is the isometry
corresponding to $n$ uses of the given channel. Recall from the proof of
Theorem~\ref{thm:up-bnd-strong-locking-cap}\ that we obtain the following
upper bound on the strong locking capacity of~$\mathcal{N}$:%
\begin{equation}
I\left(  M;B^{n}\right)  -I_{\text{acc}}\left(  M;B^{n}E^{n}\right)  +o\left(
n\right)  +n2\varepsilon^{\prime}. \label{eq:strong-locking-upper-bound}%
\end{equation}
(Recall that this bound holds for any $\left(  n,R,\varepsilon\right)
$\ strong locking protocol, with $\varepsilon^{\prime}$ a function of
$\varepsilon$ that vanishes as $\varepsilon\rightarrow0$.)\ Consider that the
information quantity $I\left(  M;B^{n}\right)  $ is upper bounded as follows:%
\begin{align*}
I\left(  M;B^{n}\right)  _{\rho}  &  \leq I\left(  M;A^{n}\right)  _{\rho}\\
&  =I\left(  MK;A^{n}\right)  _{\rho}-I\left(  K;A^{n}|M\right)  _{\rho},
\end{align*}
where the first inequality follows from quantum data processing, and the
equality follows from the chain rule for quantum mutual information. We then
find that%
\begin{align}
I\left(  MK;A^{n}\right)  _{\rho}  &  =H\left(  A^{n}\right)  _{\rho}-H\left(
A^{n}|MK\right)  _{\rho}\nonumber\\
&  =H\left(  A^{n}\right)  _{\rho}, \label{eq:mut-eq-vN-ent}%
\end{align}
where the second equality follows because the state on $A^{n}$ is a pure
coherent state when conditioned on systems $M$ and $K$.

On the other hand, we obtain a lower bound on the accessible information
$I_{\text{acc}}\left(  M;B^{n}E^{n}\right)  =I_{\text{acc}}\left(
M;A^{n}\right)  $\ by having the adversary perform heterodyne detection (a
particular measurement that is not necessarily the optimal one) on each of the
systems $A^{n}$, giving%
\begin{align}
I_{\mathrm{acc}}\left(  M;A^{n}\right)  _{\rho} &  \geq I_{\mathrm{het}%
}(M;A^{n})_{\rho}\\
&  =I_{\mathrm{het}}\left(  MK;A^{n}\right)  _{\rho}-I_{\mathrm{het}}\left(
K;A^{n}|M\right)  _{\rho},
\end{align}
where in the second line we again apply the chain rule for mutual information.
An ideal $n$-mode heterodyne measurement is described by a POVM $\{\frac
{d^{2n}\beta^{n}}{\pi^{n}}|\beta^{n}\rangle\langle\beta^{n}|\}$, where
$\beta^{n}$ is the amplitude of the $n$-mode coherent state $|\beta^{n}%
\rangle\equiv\left\vert \beta_{1}\right\rangle \cdots\left\vert \beta
_{n}\right\rangle $ and $d^{2n}\beta^{n}$ denotes the Lebesgue measure on
$\mathbb{C}^{n}$. We can then compute the heterodyne mutual information
$I_{\mathrm{het}}\left(  MK;A^{n}\right)  _{\rho}$\ as%
\[
I_{\mathrm{het}}\left(  MK;A^{n}\right)  _{\rho}=W\left(  A^{n}\right)
_{\rho}-W\left(  A^{n}|MK\right)  _{\rho}\ ,
\]
where%
\begin{equation}
W(Q)_{\sigma}=-\int\frac{d^{2n}\beta^{n}}{\pi^{n}}\langle\beta^{n}|\sigma
|\beta^{n}\rangle\log_{2}{\langle\beta}^{n}{|\sigma|\beta}^{n}{\rangle}%
\end{equation}
denotes the Wehrl entropy for a state $\sigma$ defined on system $Q$
\cite{Wehrl} and its conditional version follows in the natural way. It is
easy to see that the Wehrl entropy of an $n$-mode coherent state is equal to
$n\log_{2}\left(  e\right)  $, so we find that%
\begin{equation}
I_{\mathrm{het}}\left(  MK;A^{n}\right)  _{\rho}=W\left(  A^{n}\right)
_{\rho}-n\log_{2}\left(  e\right)  .\label{eq:heterodyne-info}%
\end{equation}
We are now in a position to derive an upper bound on
(\ref{eq:strong-locking-upper-bound}). Observe that our development above
implies that%
\begin{align}
I\left(  M;B^{n}\right)  -I_{\text{acc}}\left(  M;B^{n}E^{n}\right)   &  \leq
I\left(  MK;A^{n}\right)  _{\rho}-I\left(  K;A^{n}|M\right)  _{\rho
}\nonumber\\
&  \ \ \ \ -\left[  I_{\mathrm{het}}\left(  MK;A^{n}\right)  _{\rho
}-I_{\mathrm{het}}\left(  K;A^{n}|M\right)  _{\rho}\right]  \nonumber\\
&  \leq I\left(  MK;A^{n}\right)  _{\rho}-I_{\mathrm{het}}\left(
MK;A^{n}\right)  _{\rho}\nonumber\\
&  \leq\max_{p_{X}\left(  x\right)  }\left[  I\left(  X;A^{n}\right)
_{\omega}-I_{\mathrm{het}}\left(  X;A^{n}\right)  _{\omega}\right]
\nonumber\\
&  \leq n\ \max_{p_{X}\left(  x\right)  }\left[  I\left(  X;A\right)
_{\sigma}-I_{\mathrm{het}}\left(  X;A\right)  _{\sigma}\right]  \nonumber\\
&  =n\ \left(  \log_{2}\left(  e\right)  +\max_{p_{X}\left(  x\right)
}\left[  H\left(  A\right)  _{\sigma}-W\left(  A\right)  _{\sigma}\right]
\right)  \nonumber\\
&  \leq n\ \left[  g\left(  N_{S}\right)  -\log_{2}\left(  1+N_{S}\right)
\right]  .\label{eq:heterodyne-bound-1}%
\end{align}
The second inequality follows from data processing: $I\left(  K;A^{n}%
|M\right)  _{\rho}\geq I_{\mathrm{het}}\left(  K;A^{n}|M\right)  _{\rho}$ (the
system $M$ is classical, and performing heterodyne detection on $A^{n}$ can
only reduce the mutual information). The third inequality follows by taking a
maximization over all distributions $p_{X}\left(  x\right)  $ where
$\omega_{XA^{n}}$ is a state of the following form:%
\[
\omega_{XA^{n}}\equiv\sum_{x}p_{X}\left(  x\right)  \left\vert x\right\rangle
\left\langle x\right\vert _{X}\otimes\left\vert \alpha_{x}^{n}\right\rangle
\left\langle \alpha_{x}^{n}\right\vert _{A^{n}},
\]
such that the mean input photon number to the channel for each $x$ is $N_{S}$. The fourth
inequality follows by realizing that the difference between the mutual
information and the heterodyne information is equal to the private information
of a quantum wiretap channel in which the state $\left\vert \alpha_{x}%
^{n}\right\rangle $ is prepared for the receiver while the heterodyned version
of this state (a classical variable) is prepared for the eavesdropper. Such a
quantum wiretap channel has pure product input states (they are coherent
states) and it is degraded. Thus, we can apply
Theorem~\ref{thm:private-info-additive}\ from the appendix to show that this
private information is subadditive, in the sense that%
\[
\max_{p_{X}\left(  x\right)  }\left[  I\left(  X;A^{n}\right)  _{\omega
}-I_{\mathrm{het}}\left(  X;A^{n}\right)  _{\omega}\right]  \leq
n\ \max_{p_{X}\left(  x\right)  }\left[  I\left(  X;A\right)  _{\sigma
}-I_{\mathrm{het}}\left(  X;A\right)  _{\sigma}\right]  ,
\]
where we define the state $\sigma_{XA}$ as follows:%
\[
\sigma_{XA}\equiv\sum_{x}p_{X}\left(  x\right)  \left\vert x\right\rangle
\left\langle x\right\vert _{X}\otimes\left\vert \alpha_{x}\right\rangle
\left\langle \alpha_{x}\right\vert _{A}.
\]
The last equality follows from the observation in (\ref{eq:heterodyne-info})
and because $I\left(  X;A\right)  _{\sigma}=H\left(  A\right)  _{\sigma
}-H\left(  A|X\right)  _{\sigma}=H\left(  A\right)  _{\sigma}$ (since the
states are pure when conditioned on $X$).

We now show that the maximizing distribution for $\max_{p_{X}\left(  x\right)
}\left[  H\left(  A\right)  _{\sigma}-W\left(  A\right)  _{\sigma}\right]  $
is given by a circularly-symmetric Gaussian distribution with variance $N_{S}%
$, so that the optimal ensemble is a Gaussian ensemble of coherent states.
Indeed, let $\varrho$ be a single-mode quantum state with $\mathrm{Tr}%
[a\varrho]=0$ and $\mathrm{Tr}[a^{\dagger}a\varrho]=N_{S}$ where $a^{\dagger}$
and $a$ are creation and annihilation operators, respectively. The von Neumann
entropy is given by $H(\varrho)=-\mathrm{Tr}[\varrho\log_{2}\varrho]$. 
We show that
\begin{equation}
H(\varrho)-W(\varrho)\label{eq:h-w}%
\end{equation}
is maximized when $\varrho$ is a thermal state. Our approach is based on a
technique used in the appendix of Ref.~\cite{HSO99}, which in turn is based on
classical approaches to this problem \cite{Cover}. Let%
\begin{equation}
\widetilde{\varrho}=\frac{1}{N_{S}+1}\sum_{m=0}^{\infty}\Bigg(  \frac{N_{S}%
}{N_{S}+1}\Bigg)^m  |m\rangle\langle m|,\label{eq:gaussification}%
\end{equation}
be a thermal state with mean photon number $N_{S}$. We will show that%
\begin{equation}
H(\widetilde{\varrho}) - W(\widetilde{\varrho}) - 
\left(  H(\varrho)-W(\varrho )\right)  
\geq 0 \, , \label{eq:h-w2}%
\end{equation}
holds for any $\varrho$ with $\mathrm{Tr}[a\varrho]=0$ and $\mathrm{Tr}%
[a^{\dagger}a\varrho]=N_{S}$. 
Putting
\begin{equation}
\mathcal{Q}_{\varrho}(\beta) = \langle\beta|\varrho|\beta\rangle \, ,\label{eq:Q_function}\
\end{equation}
the left hand side of (\ref{eq:h-w2}) is equal to
\begin{align}
&  -\mathrm{Tr}[\widetilde{\varrho}\log_{2}\widetilde{\varrho}] + \mathrm{Tr}[\varrho\log_{2}\varrho]
+ \int \frac{d^{2}\beta}{\pi} \mathcal{Q}_{\widetilde{\varrho}}(\beta)\log_{2}\mathcal{Q}_{\widetilde{\varrho}}(\beta)
- \int \frac{d^{2}\beta}{\pi} \mathcal{Q}_{\varrho}(\beta)\log_{2}\mathcal{Q}_{\varrho}(\beta) \nonumber \label{eq:h-w3}\\
& = \mathrm{Tr}[\varrho(\log_{2}\varrho-\log_{2}\widetilde{\varrho})]+\mathrm{Tr}[(\varrho-\widetilde{\varrho})\log_{2}\widetilde{\varrho
}]\nonumber\\
&  \ \ \ \ \ -\left\{  \int \frac{d^{2}\beta}{\pi} \mathcal{Q}_{\varrho}(\beta)(\log_{2}\mathcal{Q}_{\varrho}(\beta)-\log_{2}\mathcal{Q}_{\tilde{\rho}}(\beta) )
+ \int \frac{d^{2}\beta}{\pi} (\mathcal{Q}_{\varrho}(\beta)-\mathcal{Q}_{\widetilde{\varrho}}(\beta))\log_{2}\mathcal{Q}_{\widetilde{\varrho}}%
(\beta)\right\}  \nonumber\\
& = D(\varrho||\widetilde{\varrho})-D(\mathcal{Q}_{\varrho}||\mathcal{Q}_{\widetilde{\varrho}})
+ \mathrm{Tr}[(\varrho-\widetilde{\varrho})\log_{2}\widetilde{\varrho}]
- \int \frac{d^{2}\beta}{\pi} (\mathcal{Q}_{\varrho}(\beta)-\mathcal{Q}_{\widetilde{\varrho}}(\beta))\log
_{2}\mathcal{Q}_{\widetilde{\varrho}}(\beta),
\end{align}
where $D(\varrho||\widetilde{\varrho})$ and $D(\mathcal{Q}_{\varrho}||\mathcal{Q}_{\widetilde
{\varrho}})$ are quantum and classical relative entropies, respectively. We
can easily show that their difference is positive by the monotonicity
property of the relative entropy. The third term is%
\begin{align}
\mathrm{Tr}[(\varrho-\widetilde{\varrho})\log_{2}\widetilde{\varrho}] &
=\mathrm{Tr}\left[  (\varrho-\widetilde{\varrho})\sum_{m=0}^{\infty}\log
_{2}\left\{  \frac{1}{N_{S}+1}\Bigg(  \frac{N_{S}}{N_{S}+1}\Bigg)
^{a^{\dagger}a}\right\}  |m\rangle\langle m|\right]
\nonumber\label{eq:third_term}\\
&  =-\log_{2}(N_{S}+1)\mathrm{Tr}[\varrho-\widetilde{\varrho}]+\log_{2}\left(
\frac{N_{S}}{N_{S}+1}\right)  \mathrm{Tr}[(\varrho-\widetilde{\varrho
})a^{\dagger}a]\nonumber\\
&  =0.
\end{align}
Similarly, the fourth term is%
\begin{align}
& \int \frac{d^{2}\beta}{\pi} \left(  \mathcal{Q}_{\varrho}(\beta)-\mathcal{Q}_{\widetilde{\varrho}}%
(\beta)\right)  \log_{2}\mathcal{Q}_{\widetilde{\varrho}}(\beta
)\nonumber\label{eq:fourth_term}\\
&  = \int \frac{d^{2}\beta}{\pi} \left(  \mathcal{Q}_{\varrho}(\beta)-\mathcal{Q}_{\widetilde{\varrho}}%
(\beta)\right)  \left(  -\log_{2}(N_{S}+1)-\frac{|\beta|^{2}}{\ln(2) (N_{S}%
+1)}\right)  \\
&  =0.
\end{align}
Note that $\mathcal{Q}_{\widetilde{\varrho}}(\beta)=\frac{1}{(N_{S}+1)}\exp\left[
-\frac{|\beta|^{2}}{N_{S}+1}\right]  $ and we used the fact that if
$\mathrm{Tr}[a^{\dagger}a\varrho]=\mathrm{Tr}[a^{\dagger}a\tau]$ then%
\[
\int d^{2}\beta\ \mathcal{Q}_{\varrho}(\beta)\ |\beta|^{2}=\int d^{2}\beta\ \mathcal{Q}_{\tau}(\beta)\ |\beta|^{2}.
\]
As a consequence, we have%
\begin{equation}
H(\widetilde{\varrho})-W(\widetilde{\varrho})-\left(  H(\varrho)-W(\varrho)\right)  =D(\varrho||\widetilde{\varrho})-D(\mathcal{Q}_{\varrho}||\mathcal{Q}_{\widetilde{\varrho}})\geq0,
\end{equation}
which completes the proof that $\max_{p_{X}\left(  x\right)  }\left[  H\left(
A\right)  _{\sigma}-W\left(  A\right)  _{\sigma}\right]  $ is optimized by a
circularly symmetric complex Gaussian distribution with variance $N_{S}$.

Finally, we can rewrite $\log_{2}\left(  e\right)  +\max_{p_{X}\left(
x\right)  }\left[  H\left(  A\right)  _{\sigma}-W\left(  A\right)  _{\sigma
}\right]  $ as $I\left(  X;A\right)  _{\sigma}-I_{\mathrm{het}}\left(
X;A\right)  _{\sigma}$ for $X$ complex Gaussian, and these information
quantities evaluate to $g\left(  N_{S}\right)  -\log_{2}\left(  1+N_{S}%
\right)  $ in such a case. By combining the bounds in
(\ref{eq:strong-locking-upper-bound}) and (\ref{eq:heterodyne-bound-1}), we
deduce the following upper bound on the rate$~R$ of any strong locking
protocol that employs coherent-state codewords with mean photon number$~N_{S}%
$:%
\[
R\leq g\left(  N_{S}\right)  -\log_{2}\left(  1+N_{S}\right)  +\frac{o\left(
n\right)  }{n}+2\varepsilon^{\prime},
\]
which converges to $g\left(  N_{S}\right)  -\log_{2}\left(  1+N_{S}\right)  $
in the limit as $n\rightarrow\infty$ and $\varepsilon\rightarrow0$.
\end{proof}

\begin{theorem}
The weak locking capacity of a pure-loss bosonic channel with transmissivity
$\eta\in\left[  0,1\right]  $ when restricting to coherent-state encodings
with mean input photon number $N_{S}$ is upper bounded by%
\[
\max\left\{  0,g\left(  \eta N_{S}\right)  -g\left(  \left(  1-\eta\right)
N_{S}\right)  \right\}  +\left[  g\left(  \left(  1-\eta\right)  N_{S}\right)
-\log_{2}\left( 
    1+\left(
1-\eta\right)N_{S}  \right)
\right]  .
\]
The term $\max\left\{  0,g\left(  \eta N_{S}\right)  -g\left(  \left(
1-\eta\right)  N_{S}\right)  \right\}  $ is equal to the private capacity of
the pure-loss bosonic channel, while the second term is limited by the bound%
\[
\left[  g\left(  \left(  1-\eta\right)  N_{S}\right)  -\log_{2}\left( 
    1+\left(
1-\eta\right)N_{S}  \right)  \right]  \leq\log_{2}\left(
e\right)  .
\]

\end{theorem}

\begin{proof}
The proof of this theorem is somewhat similar to the proof of the previous
theorem. Nevertheless, there are some important differences, and so we give
the full proof for completeness.

In the proof of Theorem~\ref{thm:WLC},\ we obtained the following upper bound
on the weak locking capacity:%
\begin{equation}
I\left(  M;B^{n}\right)  -I_{\text{acc}}\left(  M;E^{n}\right)  +o\left(
n\right)  +n2\varepsilon^{\prime}.\label{eq:WLC-bosonic-bound}%
\end{equation}
(Recall that this bound holds for any $\left(  n,R,\varepsilon\right)
$\ strong locking protocol, with $\varepsilon^{\prime}$ a function of
$\varepsilon$ that vanishes when $\varepsilon\rightarrow0$.)\ We begin by
bounding the quantity $I\left(  M;B^{n}\right)  -I_{\text{acc}}\left(
M;E^{n}\right)  $:%
\begin{align*}
&  I\left(  M;B^{n}\right)  -I_{\text{acc}}\left(  M;E^{n}\right)  \\
&  \leq I\left(  MK;B^{n}\right)  -\left[  I_{\text{het}}\left(
MK;E^{n}\right)  -I_{\text{het}}\left(  K;E^{n}|M\right)  \right]  \\
&  \leq I\left(  MK;B^{n}\right)  -I_{\text{het}}\left(  MK;E^{n}\right)
+o\left(  n\right)  \\
&  =H\left(  B^{n}\right)  -W\left(  E^{n}\right)  +n\log_{2}\left(  e\right)
+o\left(  n\right)  \\
&  =H\left(  B^{n}\right)  -H\left(  E^{n}\right)  +H\left(  E^{n}\right)
-W\left(  E^{n}\right)  +n\log_{2}\left(  e\right)  +o\left(  n\right)  \\
&  \leq n\left[  \max\left\{  0,g\left(  \eta N_{S}\right)  -g\left(  \left(
1-\eta\right)  N_{S}\right)  \right\}  \right]  \\
&  \ \ \ \ \ +n\left[  g\left(  \left(  1-\eta\right)  N_{S}\right)  -
\log_{2}\left( 
    1+\left(
1-\eta\right)N_{S}  \right)  \right]
+o\left(  n\right)  .
\end{align*}
The first inequality follows from data processing $I\left(  M;B^{n}\right)
\leq I\left(  MK;B^{n}\right)  $, the fact that $I_{\text{acc}}\left(
M;E^{n}\right)  \geq I_{\text{het}}\left(  MK;E^{n}\right)  $, and the
identity $I_{\text{het}}\left(  M;E^{n}\right)  =I_{\text{het}}\left(
MK;E^{n}\right)  -I_{\text{het}}\left(  K;E^{n}|M\right)  $. The second
inequality follows because $I_{\text{het}}\left(  K;E^{n}|M\right)  \leq
H\left(  K\right)  \leq o\left(  n\right)  $. The first equality follows from
the fact that $I\left(  MK;B^{n}\right)  =H\left(  B^{n}\right)  $ for the
pure-loss bosonic channel and from the fact that $I_{\text{het}}\left(
MK;E^{n}\right)  =W\left(  E^{n}\right)  -n\log_{2}\left(  e\right)  $. The
second equality is a simple identity. The final inequality follows because the
entropy difference $H\left(  B^{n}\right)  -H\left(  E^{n}\right)  $ is equal
to a coherent information of the $n$-use pure-loss bosonic channel. The only
relevant property of the input state for which the coherent information is
evaluated is that it has a mean photon number $N_{S}$, and so the coherent
information is always lower than $n\max\left\{  0,g\left(  \eta N_{S}\right)
-g\left(  \left(  1-\eta\right)  N_{S}\right)  \right\}  $, which is equal to
$n$ times the quantum and private capacity of this channel \cite{WPG07,WHG12}.
We also employ an argument similar to that in the previous theorem to bound
$H\left(  E^{n}\right)  -W\left(  E^{n}\right)  +n\log_{2}\left(  e\right)  $
from above by $n\left[  g\left(  \left(  1-\eta\right)  N_{S}\right)
-\log_{2}\left( 
    1+\left(
1-\eta\right)N_{S}  \right)
\right]  $. Finally, by combining the above bound with the bound in
(\ref{eq:WLC-bosonic-bound}), we deduce the following upper bound on the rate
$R$ of any weak locking protocol that employs coherent-state codewords for
transmission over a pure-loss bosonic channel:%
\[
R\leq\max\left\{  0,g\left(  \eta N_{S}\right)  -g\left(  \left(
1-\eta\right)  N_{S}\right)  \right\}  +\left[  g\left(  \left(
1-\eta\right)  N_{S}\right)  -\log_{2}\left( 
    1+\left(
1-\eta\right)N_{S}  \right)  \right]  +\frac{o\left(  n\right)  }{n}%
+2\varepsilon^{\prime},
\]
which converges to $\max\left\{  0,g\left(  \eta N_{S}\right)  -g\left(
\left(  1-\eta\right)  N_{S}\right)  \right\}  +\left[  g\left(  \left(
1-\eta\right)  N_{S}\right)  -\log_{2}\left( 
    1+\left(
1-\eta\right)N_{S}  \right)  \right]  $ in the limit as $n\rightarrow\infty$ and
$\varepsilon\rightarrow0$.
\end{proof}

\begin{remark}
Given that the private capacity of a pure-loss bosonic channel with mean input
photon number $N_{S}$ is equal to $\max\left\{  0,g\left(  \eta N_{S}\right)
-g\left(  \left(  1-\eta\right)  N_{S}\right)  \right\}  $, the above theorem
implies a strong limitation on the weak locking capacity of a pure-loss
bosonic channel when restricting to coherent-state encodings with mean input
photon number $N_{S}$. That is, the weak locking capacity when restricting to
coherent-state encodings cannot be more than 1.45 bits larger than the channel's
private capacity.
\end{remark}

\begin{remark}
These bounds apply in particular to channels that use a coherent-state locking protocol in which
there is a fixed codebook $\left\{  \left\vert \alpha^{n}\left(  m\right)
\right\rangle \right\}  $ and the coherent states are encrypted according to
passive mode transformations $U_{k}$ that transform $n$-mode\ coherent states
as $\left\vert \alpha^{n}\right\rangle \rightarrow|\widetilde{U}_{k}\alpha
^{n}\rangle$, where $\widetilde{U}_{k}\alpha^{n}$ is understood to be a label
for a coherent state vector with the following complex amplitudes:%
\begin{equation}%
\begin{bmatrix}
\widetilde{U}_{k}^{\left(  1,1\right)  } & \cdots & \widetilde{U}_{k}^{\left(
1,n\right)  }\\
\vdots & \ddots & \vdots\\
\widetilde{U}_{k}^{\left(  n,1\right)  } & \cdots & \widetilde{U}_{k}^{\left(
n,n\right)  }%
\end{bmatrix}%
\begin{bmatrix}
\alpha_{1}\\
\vdots\\
\alpha_{n}%
\end{bmatrix}
.\label{eq:passive-mode-trans}%
\end{equation}

\end{remark}

\begin{remark}
If the coherent-state locking protocol consists of passive mode
transformations (as defined above) for the encryption and the receiver
performs heterodyne detection to recover the message after decrypting with a
passive mode transformation, then the strong locking capacity of a channel
using such a scheme
is equal to zero. This result follows because such a scheme has a classical
simulation---passive mode transformations commute with heterodyne detection in
such a way that heterodyne detection can be performed first followed by a
classical postprocessing of the measurement data with a matrix multiplication
as in (\ref{eq:passive-mode-trans}). That is, the decoding in such a scheme is
equivalent to first performing key-independent heterodyne detection
measurements followed by classical post-processing of the key and the
measurement results. Thus, Theorem~\ref{thm:classical-sim}\ applies
so that the strong locking capacity of a channel using such a
scheme is equal to zero. However, this theorem does not apply if the receiver performs
photodetection because passive mode transformations do not commute with
such a measurement.
\end{remark}

From our upper bounds on the locking capacity of channels restricted to coherent-state
encodings, it is clear that there are strong limitations on the rates that are
achievable when employing bright coherent states. That is, it clearly would
not be worthwhile to invest a large mean input photon number per transmission
given the above limitations on locking capacity that are independent of the
photon number. In spite of this result, it might be possible to achieve
interesting locking rates with weak coherent states, but we should keep in
mind that the above bounds were derived by considering the information that an
adversary can gain by performing heterodyne detection---the information of the
adversary can only increase if she performs a better measurement.
Nevertheless, we can determine values of the mean input photon number $N_{S}$
such that the difference $g\left(  N_{S}\right)  -\log_{2}\left(
1+N_{S}\right)  $ becomes relatively large. By considering $N_{S}\ll1$, we
find the following expansions of $g(N_{S})$ and $\log_{2}{(1+N_{S})}$,
respectively:%
\begin{align}
g(N_{S})  &  \approx\left(  -N_{S}\ln{N_{S}}+N_{S}+\frac{N_{S}^{2}}{2}\right)
\log_{2}\left(  {e}\right)  \,,\\
\log_{2}{(1+N_{S})}  &  \approx\left(  N_{S}-\frac{N_{S}^{2}}{2}\right)
\log_{2}\left(  {e}\right)  \,,
\end{align}
so that the difference $g\left(  N_{S}\right)  -\log_{2}\left(  1+N_{S}%
\right)  \approx\left[  -N_{S}\ln{N_{S}+}N_{S}^{2}\right]  \log_{2}\left(
e\right)  $ for $N_{S}\ll1$. Indeed, in the limit as $N_{S}\rightarrow0$, we
find that the relative ratio of our upper bound on the strong locking capacity
of a coherent-state protocol to the classical capacity $g\left(  N_{S}\right)
$\ of the noiseless bosonic channel approaches one:%
\[
\lim_{N_{S}\rightarrow0}\frac{g(N_{S})-\log_{2}{(1+N_{S})}}{g(N_{S})} = 1,
\]
so that there is some sense in which the rate at which we can lock information
becomes similar to the rate at which we can insecurely communicate information
if the bound $g(N_{S})-\log_{2}{(1+N_{S})}$ is in fact achievable. However,
this remains an important open question.

\section{One-shot PPM coherent-state locking protocol}\label{weakQEM}

In spite of the previous section's limitations on the locking capacity of
coherent-state protocols, we still think it is interesting to explore what
kind of locking protocols are possible using coherent-state encodings. To this
end, we now discuss a one-shot strong locking protocol (in the sense of
Definition~\ref{def:one-shot-locking-cap}) that employs a coherent-state encoding.
For simplicity, we consider the case of a noiseless channel,
with the generalization to the pure-loss bosonic channel
and weak locking being straightforward.

An explicit scheme for locking using weak coherent states can be obtained by
analogy with the PPM encryption presented in Ref.~\cite{QEM} and reviewed in
Section~\ref{unaryQEM}. Similar to the single-photon scheme, to encode a
message $m$ Alice prepares an $n$-mode coherent state $|\alpha_{m}\rangle$
which is a tensor product of a single-mode coherent state of amplitude~$\alpha
$ on the $m$th mode and the vacuum on the remaining $n-1$ modes:
\[
|\alpha_{m}\rangle\equiv|0\rangle_{1}\ldots|0\rangle_{m-1}|\alpha\rangle
_{m}|0\rangle_{m+1}\ldots|0\rangle_{n}.
\]
(Notice that, as in the single-photon case, the PPM encoding is highly
inefficient in terms of number of modes, as it encodes only  $\log_{2}{n}$ bits
into $n$ bosonic modes.) Let us fix $N_{\mathrm{tot}}=|\alpha|^{2}$ to be the
total mean number of photons involved in the protocol, and
\begin{equation}
N_{S}\equiv\frac{N_{\mathrm{tot}}}{n}%
\end{equation}
to be the mean photon number per mode. Before sending anything to Bob, Alice
encrypts a message by applying a unitary selected uniformly at random
(according to the shared secret key) from a set of $\left\vert \mathcal{K}%
\right\vert $ $n$-mode linear-optical passive transformations. If the unitary
$U_{k}$ is used, then the final state is the $n$-mode coherent state
\begin{equation}
U_{k}|\alpha_{m}\rangle=\bigotimes_{m^{\prime}=1}^{n} | \widetilde{U}%
_{k}^{\left(  m^{\prime},m\right)  } \alpha\rangle\,.
\end{equation}
Bob, who knows which unitary has been chosen by Alice, applies the inverse
transformation and performs photodetection on the received modes. He will
detect (one or more) photons only in the $m$th mode, hence successfully
decrypting the message in case of a detection.

Different from the single-photon architecture of Ref.~\cite{QEM}, there is a
non-zero probability that Bob's detector does not click. Analogous to the case
of the single-photon locking protocol in the presence of loss, if no photon is
detected, Bob may use a public classical communication channel to ask Alice to resend,
yielding
\begin{equation}
I_{\mathrm{acc}}(M;KQ)_{\rho} = N_\mathrm{tot} \log_{2}{n}\,.
\end{equation}
However, the same observations from Section~\ref{unaryQEM}
apply here. That is, locking is only known
to be secure when the message distribution is uniform, and this is certainly not the case
for a feedback-assisted scheme unless Eve attacks each PPM block independently. If she attacks
collectively, then it is necessary for Alice and Bob to exploit an amount of key necessary
to ensure that the message distribution is uniform.

Assuming that Eve independently attacks each block that she receives,
we have to evaluate her accessible information with respect to the following state:
\begin{equation}
\rho_{MKQ} = \frac{1}{n \left\vert
\mathcal{K}\right\vert }\sum_{m,k}|m,k\rangle\langle m,k|_{MK}\otimes\left(
U_{k}|\alpha_{m}\rangle\langle\alpha_{m}|U_{k}^{\dag}\right)  _{Q}\,.
\end{equation}
If the set of $n$-mode unitaries are selected uniformly at random according to
the Haar measure, one might expect that such a set of unitaries scrambles
phase information of $\alpha$ so that it is not accessible to Eve. Thus, a
presumably clever strategy for Eve is to perform a measurement that commutes
with the total number of photons. Such a POVM has elements%
\begin{equation}
\{|0\rangle\langle0|,\{\mu_{y}^{(1)}|\phi_{y}^{(1)}\rangle\langle\phi
_{y}^{(1)}|\}_{y},\{\mu_{y}^{(2)}|\phi_{y}^{(2)}\rangle\langle\phi_{y}%
^{(2)}|\}_{y},\dots\}, \label{eq:photon-number-POVM-elements}%
\end{equation}
where $|0\rangle$ is the $n$-mode vacuum, and for any $k\geq1$ each vector
$|\phi_{y}^{(k)}\rangle$ belongs to the $k$-photon subspace. This suboptimal
measurement allows Eve to achieve a mutual information $I_{\mathrm{num}%
}(M;Q)_{\rho}$ such that%
\begin{equation}
I_{\mathrm{num}}(M;Q)_{\rho}\leq I_{\mathrm{acc}}(M;Q)_{\rho}.
\end{equation}
For small values of $N_{\mathrm{tot}}\ll1$, the probability of having more
than one photon is of order $N_{\mathrm{tot}}^{2}$. We can hence argue that
for $N_{\mathrm{tot}}^{2}\ll1$, the main contribution to $I_{\mathrm{num}%
}(M;Q)_{\rho}$ comes from POVM elements in
(\ref{eq:photon-number-POVM-elements}) with $k=0,1$ and that the contribution
of those with $k\geq2$ is, in the worst case, of order $N_{\mathrm{tot}}%
^{2}\log_{2}{n}$. Noticing that the projection of the coherent state
$U_{k}|\alpha_{m}\rangle$ in the subspace spanned by the vacuum and the
single-photon subspace is
\begin{equation}
e^{-\frac{|\alpha|^{2}}{2}}(|0\rangle+\alpha U_{k}|m\rangle)=e^{-\frac
{|\alpha|^{2}}{2}}(|0\rangle+\alpha\sum_{m^{\prime}} \widetilde{U}%
_{k}^{\left(  m^{\prime},m\right)  } |m^{\prime}\rangle)\,,
\end{equation}
where $|m^{\prime}\rangle$ is the state of a single photon on the $m^{\prime}%
$th mode, a straightforward calculation leads to the following expression for
the lower bound $I_{\mathrm{num}}(M;Q)_{\rho}$:
\begin{equation}
I_{\mathrm{num}}(M;Q)_{\rho}=N_{\mathrm{tot}}\left[  \log_{2}{n}%
-\min_{\mathcal{M}_{E}^{(1)}}\sum_{y}\frac{\mu_{y}^{(1)}}{n
\left\vert \mathcal{K}\right\vert }\sum_{k}%
H(q_{yk})\right]  +O(N_{\mathrm{tot}}^{2}\log_{2}{n})\,, \label{Inum}%
\end{equation}
where the optimization is over the POVM $\mathcal{M}_{E}^{(1)}$ defined on the
single-photon subspace, with elements $\{\mu_{y}^{(1)}|\phi_{y}^{(1)}%
\rangle\langle\phi_{y}^{(1)}|\}_{i}$, and $q_{yk}^{m}=|\langle\phi_{y}%
^{(1)}|U_{k}|m\rangle|^{2}$.

The expression in square brackets in (\ref{Inum}) is formally the same as that
in (\ref{acc}). We can hence bound $I_{\mathrm{num}}(M;Q)_{\rho}$ using the
results of Ref.~\cite{FHS11}. It follows that there exist choices of
$\left\vert \mathcal{K}\right\vert $ $n$-mode passive linear-optical unitaries
with
\begin{equation}
\log_{2}{|\mathcal{K}|} = 4\log_{2}{(\varepsilon^{-1})} + O(\log_2\log_{2}{(\varepsilon^{-1})})\,,
\end{equation}
such that%
\begin{equation}
I_{\mathrm{num}}(M;Q)_{\rho} \leq \varepsilon N_{\mathrm{tot}} \log_{2}{n} + O(N_{\mathrm{tot}}^{2}\log_{2}{n})\,,
\end{equation}
with $\varepsilon$ arbitrarily small if $n$ is large enough. 

Clearly, the security condition $r_{1}\ll1$ can be satisfied only if
$I_{\mathrm{num}}(M;Q)_{\rho}\ll I_{\mathrm{acc}}(M;QK)_{\rho}$.
A necessary condition for $r_{1}\ll1$ to hold is
\begin{equation}
\varepsilon + O(N_{\mathrm{tot}}) \ll 1 \,,
\label{sec-cond}%
\end{equation}
which can be fulfilled in the case of weak coherent states, where $N_{\mathrm{tot}} \ll 1$.
In this regime, the key-efficiency condition $r_{2}<1$ can be satisfied only if
\begin{equation}
4\log_{2}{(\varepsilon^{-1} )} < N_{\mathrm{tot}} \log_{2}{n} \,.
\end{equation}
This implies that the value of $N_{\mathrm{tot}}$ has to be in the range
\begin{equation}\label{effweak}
1\gg N_{\mathrm{tot}} > \frac{4\log_{2}{(\varepsilon^{-1})}}{\log_{2} {n}} \, .
\end{equation}

In conclusion, this weak coherent state PPM protocol is analogous to the single-photon one in
the presence of linear loss. In principle the condition in (\ref{effweak}) can
always be fulfilled for $n$ large enough, yet the minimum value of $n$ increases exponentially with decreasing 
key-efficiency ratio $r_2$ and with decreasing $N_\mathrm{tot}$.

\section{Conclusion}

\label{end}In this paper, we formally defined the locking capacity of a
quantum channel in order to establish a framework for understanding the locking effect
in the presence of noise. We can distinguish between a weak locking capacity
and a strong one, the difference being whether the adversary has access to the
environment of the channel or to its input. We related these locking
capacities to other well known capacities from quantum Shannon theory such as
the quantum, private, and classical capacity. The existence of the
FHS\ locking protocol \cite{FHS11}\ establishes that both the weak and strong
capacity locking capacities are not smaller than the quantum capacity, while
the weak locking capacity is not smaller than the private capacity because a
private communication protocol always satisfies the demands of a weak locking
protocol. Furthermore, the classical capacity is a trivial upper bound on both
locking capacities. We also proved that the strong locking capacity is equal
to zero whenever a locking protocol has a classical simulation and that both
locking capacities are equal to zero for an entanglement-breaking channel.
This latter result demonstrates that a channel should have some ability to
preserve entanglement in order for non-zero locking rates to be achievable.
Moreover, we found a class of channels for which the weak locking capacity is equal
to both the private capacity and the quantum capacity.

As an important application, we considered the case of the pure-loss bosonic
channel and the locking capacities for channels restricted to coherent-state
encodings. We note that a particular example of such a protocol is the
$\alpha\eta$ protocol (also known as $Y00$) \cite{Y00}.
We found limitations of the locking capacity for these
coherent-state schemes:\ the strong locking capacity of any channel is not
larger than $\log_{2}\left(  e\right)  $ locked bits per channel use while the
weak locking capacity of the pure-loss bosonic channel is
not larger than the sum of its private capacity and $\log_{2}\left(
e\right)  $ locked bits per channel use. If the scheme exploits passive mode
transformations and the receiver uses heterodyne detection, the restrictions
are as severe as they can be:\ the strong locking capacity is equal to zero
because there is a classical simulation of such a protocol.

As a final contribution, we discussed locking schemes that exploit weak
coherent states and that might be physically implementable. They are similar to the
single-photon quantum enigma machine (QEM) of Ref.~\cite{QEM}, with the
exception that information is encoded by pulse position modulation (PPM) of a
coherent state of a given amplitude~$\alpha$, with $|\alpha|^{2}\ll1$, over
$n$ modes. The necessary conditions for security and key efficiency of this
scheme are qualitatively equivalent to that of the single-photon QEM.

The realization of a proof of principle demonstration of a quantum enigma
machine is a tremendous experimental challenge. The main difficulty to
overcome concerns the scaling of the physical resources required for key
efficient encryption. Notice that for single-photon PPM encoding, $n$ optical
modes are needed to encode $\log_{2}{n}$ bits, while the required secret key
has length of the order of $\log_{2}{\log_{2}{n}}$. As a consequence, the
number of modes increases very quickly if one requires small values of the key
efficiency ratio $r_{2}$, as defined in (\ref{r2}).
Although keeping the same scaling law, the resources required for a key
efficient single-photon QEM become even more demanding when one introduces loss in the
single-photon scheme and when one moves from the single-photon PPM to the weak
coherent-state PPM. On the other hand, in the case of a weak coherent-state
QEM, one can trade off the increase in the number of modes (and/or the
reduction in the key efficiency level) with the fact that coherent states can
be prepared deterministically and are much easier to handle than single-photon
states. However, it seems that one has to go beyond PPM encoding to overcome
the key efficiency limitations.

There are many open questions to consider going forward from here. Perhaps the
most pressing question is to determine a formula\ that serves as a good lower
bound on the locking capacities (that is, one would need to demonstrate
locking protocols with nontrivial achievable rates according to the
requirements stated in Definitions~\ref{def:weak-lock-protocol},
\ref{def:strong-lock-protocol}, and \ref{def:achievable-rate-locking}). One
might suspect that the formulas given in Theorems~\ref{thm:WLC} and
\ref{thm:up-bnd-strong-locking-cap}\ are in fact achievable, but it is not
clear to us if this is true. Furthermore, it is important to determine if
there is an example of channel (perhaps many?) for which its weak locking
capacity is strictly larger than its private classical capacity, and
similarly, if there is a channel for which its strong locking capacity is
strictly larger than its quantum capacity. We also suspect that the locking
capacity is non-additive, as is the case for other capacities in quantum
Shannon theory \cite{science2008smith,LWZG09,SS09}. If it is the case that the
50\% quantum erasure channel has a weak locking capacity equal to zero, then
it immediately follows from the results of Ref.~\cite{SS09}\ and our
operational bounds in (\ref{eq:weak-lock-bounds}) and
(\ref{eq:strong-lock-bounds}) that both the weak and strong locking capacities
are non-additive.

Another intriguing question is the relationship between the
strong locking and quantum identification capacities \cite{Win13}
of a quantum channel. Both seem to involve a weak form of
coherent data transmission from a sender to receiver. FHS even
established an explicit connection between locking using
unitary encodings and quantum identification over a channel
built out of the inverses of those unitaries \cite{FHS11}. It is tempting
to speculate that the single-letter formula for the amortized
quantum indentification capacity found in Ref.~\cite{HW12} could
thereby be recruited as a tool to study the locking capacity.

\bigskip

\textbf{Acknowledgements}. We are grateful to Omar Fawzi, Graeme Smith, and
Andreas Winter for helpful discussions. This research was supported by the
DARPA\ Quiness Program through US\ Army Research Office award W31P4Q-12-1-0019.
PH was supported by the Canada Research Chairs program, CIFAR, NSERC and ONR through grant N000140811249.

\appendix

\section{Private capacity of degraded quantum wiretap channels}

This appendix contains a proof that the private capacity of a degraded quantum
wiretap channel when restricted to product-state encodings is single-letter.
Also, we show by an appeal to Hastings' counterexample to the additivity
conjecture \cite{H09}\ that there exists two quantum wiretap channels that are
degraded but nevertheless have non-additive private information. This latter
result provides a simple answer to a question that has been open since the
introduction of weakly degradable channels \cite{CG06}.

A quantum wiretap channel is defined as a completely positive trace-preserving
map $\mathcal{N}_{A\rightarrow BE}$ from an input system $A$ to a legitimate
receiver's system $B$ and an eavesdropper's system $E$. Such a map has an
isometric extension $U_{A\rightarrow BEF}$ with the property that%
\[
\mathcal{N}_{A\rightarrow BE}\left(  \rho\right)  =\text{Tr}_{F}\left\{
U_{A\rightarrow BEF}\rho U_{A\rightarrow BEF}^{\dag}\right\}  .
\]
The private capacity of a quantum wiretap channel is given by
\cite{ieee2005dev,1050633}%
\[
\lim_{n\rightarrow\infty}\frac{1}{n}P\left(  \mathcal{N}_{A\rightarrow
BE}^{\otimes n}\right)  ,
\]
where $P\left(  \mathcal{N}_{A\rightarrow BE}\right)  $ is the private
information, defined as%
\begin{equation}
P\left(  \mathcal{N}_{A\rightarrow BE}\right)  \equiv\max_{\left\{
p_{X}\left(  x\right)  ,\rho_{x}\right\}  }I\left(  X;B\right)  _{\rho
}-I\left(  X;E\right)  _{\rho}, \label{eq:private-info-formula}%
\end{equation}
with the entropies taken with respect to the following classical-quantum
state:%
\begin{equation}
\rho_{XBE}\equiv\sum_{x}p_{X}\left(  x\right)  \left\vert x\right\rangle
\left\langle x\right\vert _{X}\otimes\mathcal{N}_{A\rightarrow BE}\left(
\rho_{x}\right)  . \label{eq:private-code-state}%
\end{equation}
Such a wiretap channel is degraded if there exists a degrading map
$\mathcal{D}_{B\rightarrow E}$ such that%
\[
\mathcal{D}_{B\rightarrow E}\circ\mathcal{N}_{A\rightarrow B}=\mathcal{N}%
_{A\rightarrow E}.
\]

\begin{theorem}
The private information formula in (\ref{eq:private-info-formula}) for two
degraded quantum wiretap channels is generally non-additive. That is, there
exists degraded quantum wiretap channels $\mathcal{N}_{1}$ and $\mathcal{N}%
_{2}$ such that%
\[
P\left(  \mathcal{N}_{1}\otimes\mathcal{N}_{2}\right)  >P\left(
\mathcal{N}_{1}\right)  +P\left(  \mathcal{N}_{2}\right)  .
\]

\end{theorem}

\begin{proof}
This result follows by exploiting the counterexample of Hastings \cite{H09}
for the Holevo information formula. Let $\mathcal{M}_{1}$ and $\mathcal{M}%
_{2}$ be the channels from Hastings' counterexample, i.e., they satisfy%
\begin{equation}
\chi\left(  \mathcal{M}_{1}\otimes\mathcal{M}_{2}\right)  >\chi\left(
\mathcal{M}_{1}\right)  +\chi\left(  \mathcal{M}_{2}\right)  ,
\label{eq:hastings-inequality}%
\end{equation}
where $\chi\left(  \mathcal{N}\right)  $ is the Holevo information of a
channel $\mathcal{N}$. We then construct our quantum wiretap channels
$\mathcal{M}_{1}$ and $\mathcal{M}_{2}$ as $\mathcal{N}_{1}\left(
\rho\right)  =\mathcal{M}_{1}\left(  \rho\right)  \otimes\sigma_{E}$ and
$\mathcal{N}_{2}\left(  \rho\right)  =\mathcal{M}_{2}\left(  \rho\right)
\otimes\sigma_{E}$. Both channels are obviously degraded wiretap channels
because the channel to the environment simply prepares a constant state
$\sigma_{E}$. Also, there is no dependence of the environment's output on the
input state, so that the private informations of these channels reduce to
Holevo informations:%
\begin{align*}
P\left(  \mathcal{N}_{1}\otimes\mathcal{N}_{2}\right)   &  =\chi\left(
\mathcal{M}_{1}\otimes\mathcal{M}_{2}\right)  ,\\
P\left(  \mathcal{N}_{1}\right)   &  =\chi\left(  \mathcal{M}_{1}\right)  ,\\
P\left(  \mathcal{N}_{2}\right)   &  =\chi\left(  \mathcal{M}_{2}\right)  .
\end{align*}
Thus, the inequality in the statement of the theorem follows from
(\ref{eq:hastings-inequality}).
\end{proof}

\begin{theorem}
\label{thm:private-info-additive}The private information of a degraded quantum
wiretap channel is additive when restricted to product state encodings.
\end{theorem}

\begin{proof}
First, consider that we can always restrict the optimization in the private
information formula to be taken over pure input states whenever the quantum
wiretap channel $\mathcal{N}_{A\rightarrow BE}$\ is degraded. Indeed, consider
the extension state%
\begin{equation}
\rho_{XYBE}\equiv\sum_{x,y}p_{X}\left(  x\right)  p_{Y|X}\left(  y|x\right)
\left\vert x\right\rangle \left\langle x\right\vert _{X}\otimes\left\vert
y\right\rangle \left\langle y\right\vert _{Y}\otimes\mathcal{N}_{A\rightarrow
BE}\left(  \psi_{x,y}\right)  , \label{eq:ext-state}%
\end{equation}
where we are using a spectral decomposition for each state $\rho_{x}$:%
\[
\rho_{x}=\sum_{y}p_{Y|X}\left(  y|x\right)  \psi_{x,y}.
\]
Thus, the state in (\ref{eq:private-code-state}) is a reduction of the state
in (\ref{eq:ext-state}). Now consider that%
\begin{align*}
I\left(  X;B\right)  _{\rho}-I\left(  X;E\right)  _{\rho}  &  =I\left(
XY;B\right)  -I\left(  XY;E\right)  -\left[  I\left(  Y;B|X\right)  -I\left(
Y;E|X\right)  \right] \\
&  \leq I\left(  XY;B\right)  -I\left(  XY;E\right) \\
&  \leq P\left(  \mathcal{N}_{A\rightarrow BE}\right)  .
\end{align*}
The first equality is from the chain rule for mutual information. The first
inequality follows by exploiting the degrading condition and from the fact
that $X$ is classical. The final inequality follows by considering $XY$ as a
joint classical system, so that the private information of the channel can
only be larger than $I\left(  XY;B\right)  -I\left(  XY;E\right)  $.

Now consider an isometric extension $U_{A\rightarrow BEF}$\ of a quantum
wiretap channel $\mathcal{N}_{A\rightarrow BE}$. By using the fact that the
private information is optimized for pure state ensembles, we can always
rewrite it as%
\begin{align}
I\left(  X;B\right)  -I\left(  X;E\right)   &  =H\left(  B\right)  -H\left(
E\right)  -H\left(  B|X\right)  +H\left(  E|X\right) \nonumber\\
&  =H\left(  B\right)  -H\left(  E\right)  -H\left(  B|X\right)  +H\left(
BF|X\right) \nonumber\\
&  =H\left(  B\right)  -H\left(  E\right)  +H\left(  F|BX\right)  ,
\label{eq:useful-identity-private-info-wiretap}%
\end{align}
where in the second line we used the fact that $H\left(  E|X\right)  =H\left(
BF|X\right)  $ for pure-state ensembles.

Now we show the additivity property for product-state ensembles. Consider the
following state on which we evaluate information quantities:%
\[
\sigma_{XB_{1}E_{1}F_{1}B_{2}E_{2}F_{2}}\equiv\sum_{x}p_{X}\left(  x\right)
\left\vert x\right\rangle \left\langle x\right\vert _{X}\otimes U_{A_{1}%
\rightarrow B_{1}E_{1}F_{1}}\left(  \phi_{x}\right)  \otimes U_{A_{2}%
\rightarrow B_{2}E_{2}F_{2}}\left(  \psi_{x}\right)  ,
\]
where we are restricting the signaling states to be product and without loss
of generality we can take them to be pure as shown above. Consider that%
\begin{align*}
&  I\left(  X;B_{1}B_{2}\right)  _{\sigma}-I\left(  X;E_{1}E_{2}\right)
_{\sigma}\\
&  =H\left(  B_{1}B_{2}\right)  _{\sigma}-H\left(  E_{1}E_{2}\right)
_{\sigma}+H\left(  F_{1}F_{2}|B_{1}B_{2}X\right) \\
&  =H\left(  B_{1}\right)  _{\sigma}+H\left(  B_{2}\right)  _{\sigma}-H\left(
E_{1}\right)  _{\sigma}-H\left(  E_{2}\right)  _{\sigma}-\left[  I\left(
B_{1};B_{2}\right)  _{\sigma}-I\left(  E_{1};E_{2}\right)  _{\sigma}\right]
+H\left(  F_{1}F_{2}|B_{1}B_{2}X\right) \\
&  \leq H\left(  B_{1}\right)  _{\sigma}+H\left(  B_{2}\right)  _{\sigma
}-H\left(  E_{1}\right)  _{\sigma}-H\left(  E_{2}\right)  _{\sigma}+H\left(
F_{1}|B_{1}X\right)  +H\left(  F_{2}|B_{2}X\right) \\
&  =\left[  I\left(  X;B_{1}\right)  -I\left(  X;E_{1}\right)  \right]
+\left[  I\left(  X;B_{2}\right)  -I\left(  X;E_{2}\right)  \right]  .
\end{align*}
The first equality follows from the identity in
(\ref{eq:useful-identity-private-info-wiretap}). The second equality follows
from entropy identities. The first inequality follows from the degraded
wiretap channel assumption, so that $I\left(  B_{1};B_{2}\right)  _{\sigma
}-I\left(  E_{1};E_{2}\right)  _{\sigma}\geq0$ and by applying strong
subadditivity of entropy \cite{LR73}\ three times to get\ that $H\left(
F_{1}F_{2}|B_{1}B_{2}X\right)  \leq H\left(  F_{1}|B_{1}X\right)  +H\left(
F_{2}|B_{2}X\right)  $. The last equality follows from the identity in
(\ref{eq:useful-identity-private-info-wiretap}) and the fact that we are
restricting to product-state signaling ensembles.
\end{proof}

\bibliographystyle{plain}
\bibliography{Ref}

\end{document}